\documentclass{amsart}[11pt]
\usepackage{amssymb,amsmath,amsthm}
\usepackage[left=3.7cm,right=3.7cm,top=1.5cm,bottom=2.5cm]{geometry}
\usepackage[T1]{fontenc}   
\usepackage{graphicx}
\usepackage{float}
\usepackage{booktabs}
\usepackage{array}
\usepackage{hyperref}
\usepackage{mathtools}
\usepackage{xcolor}
\usepackage{enumerate}
\usepackage{bbm}
\usepackage{todonotes}

\usepackage{array}
\newcolumntype{H}{>{\setbox0=\hbox\bgroup}c<{\egroup}@{}}

\numberwithin{equation}{section}

\newtheorem{theorem}{Theorem}
\numberwithin{theorem}{subsection}
\newtheorem{proposition}{Proposition}
\numberwithin{proposition}{subsection}
\newtheorem{definition}{Definition}
\numberwithin{definition}{subsection}
\newtheorem{lemma}{Lemma}
\numberwithin{lemma}{subsection}
\newtheorem*{remark}{Remark}

\newtheorem*{example}{Example}
\newtheorem{assumption}{Assumption}
\numberwithin{assumption}{subsection}

\newcommand{\Xx}{\mathcal{X}}
\newcommand{\Yy}{\mathcal{Y}}
\newcommand{\Aa}{\mathcal{A}}
\newcommand{\Cc}{\mathcal{C}}
\newcommand{\Dd}{\mathcal{D}}
\newcommand{\Ff}{\mathcal{F}}
\newcommand{\Gg}{\mathcal{G}}
\newcommand{\Ll}{\mathcal{L}}

\newcommand{\HH}{\mathbb{H}}
\newcommand{\PP}{\mathbb{P}}
\newcommand{\QQ}{\mathbb{Q}}
\newcommand{\RR}{\mathbb{R}}
\newcommand{\NN}{\mathbb{N}}
\newcommand{\Lf}{\mathfrak{L}}

\newcommand{\R}{\RR}
\newcommand{\EE}{\mathbb{E}}

\newcommand{\VV}{\mathbb{V}}
\newcommand{\ct}{\circ_{t}}
\newcommand{\cT}{\circ_{T}}
\newcommand{\Wf}{\boldsymbol{\mathrm{W}}}

\newcommand{\Yf}{\boldsymbol{\mathrm{Y}}}
\newcommand{\vrho}{\boldsymbol{\mathrm{\varrho}}}
\newcommand{\rrho}{\overline{\rho}}
\newcommand{\Wp}{W^{\perp}}
\newcommand{\brho}{\overline{\rho}}
\newcommand{\D}{\mathrm{d}}
\newcommand{\E}{\mathrm{e}}
\newcommand{\II}{\mathrm{I}}
\newcommand{\IIV}{\mathrm{I}^{V}}
\newcommand{\IIW}{\mathrm{I}^{W}}
\newcommand{\IIX}{\mathrm{I}^{X}}
\newcommand{\IIXV}{\mathrm{I}^{X,V}}

\newcommand{\LDP}{\mathrm{LDP}}
\newcommand{\eps}{\varepsilon}
\newcommand{\seps}{\sqrt{\eps}}

\newcommand{\ww}{\boldsymbol{\mathrm{w}}}
\newcommand{\yy}{\boldsymbol{\mathrm{y}}}
\newcommand{\hh}{\boldsymbol{\mathrm{h}}}
\newcommand{\xx}{\boldsymbol{\mathrm{x}}}
\newcommand{\hm}{\mathrm{h}}
\newcommand{\xm}{\mathrm{x}}
\newcommand{\ind}{1\hspace{-2.1mm}{1}} 

\newcommand{\half}{\frac{1}{2}}

\DeclareMathOperator*{\argmax}{arg\,max}

\begin{document}

\title{Large and moderate deviations for importance sampling in the Heston model}
\author{Marc Geha}
\address{ORFE Department, Princeton University}
\email{mgeha@princeton.edu}
\author{Antoine Jacquier}
\address{Department of Mathematics, Imperial College London, and Alan Turing Institute}
\email{a.jacquier@imperial.ac.uk}
\author{\v{Z}an \v{Z}uri\v{c}}
\address{Department of Mathematics, Imperial College London}
\email{z.zuric19@imperial.ac.uk}
\thanks{AJ acknowledges financial support from the EPSRC/T032146 grant. \v{Z}\v{Z} is supported by the EPSRC/S023925 CDT in Mathematics of Random Systems: Analysis, Modelling and Simulation.}

\date{\today}
\maketitle

\begin{abstract}
We provide a detailed importance sampling analysis for variance reduction in stochastic volatility models.
The optimal change of measure is obtained using a variety of results from large and moderate deviations: small-time, large-time, small-noise.
Specialising the results to the Heston model, we derive many closed-form solutions, 
making the whole approach easy to implement.
We support our theoretical results with  a detailed numerical analysis of the variance reduction gains.
\end{abstract}

\tableofcontents

\section{Introduction and general overview}

\subsection{Introduction}
Monte Carlo simulation is the  standard (if not the only) technique for most numerical problems in stochastic modelling.
It has a long history and has been successfully applied in many fields, such as biology~\cite{manly2018randomization}, statistical Physics~\cite{binder2012monte}, Finance~\cite{glasserman2004monte} among others.
The default order of magnitude for the variance of the estimator is $\mathcal{O}(N^{-1/2})$ with~$N$ the number of sample paths.
It has long been recognised though that several tricks achieve lower variance with equivalent (hopefully zero) bias; among those antithetic variables and importance sampling have become ubiquitous.
We focus on the latter, for which large and moderate deviations (LDP and MDP) provide closed-form formulae, 
making their applications pain-free and without additional computer costs.

The first attempt to reduce the variance of a Monte Carlo estimator based on asymptotics probably originated, rather heuristically, in~\cite{Siegmund1976}. 
This was made rigorous later by Glasserman and Wang~\cite{Glasserman1997} who also highlighted pitfalls of the method and later by Dupuis and Wang~\cite{dupuis2004importance}, who provided clear explanations on the trade-off between asymptotic approximations and the restrictions they entail on the induced change of measure.
Guasoni and Robertson~\cite{Guasoni2007} 
put this into practice for out-the-money path-dependent options in the Black-Scholes models, 
and Robertson~\cite{Robertson2010} developed a thorough analysis for the Heston model using sample path large deviations.
This is our starting point,
and the goal of our current enterprise is to analyse different asymptotic regimes 
(small-time, large-time, small-noise),
both in the large deviations and in the moderate deviations regimes, in the Heston model 
and to show how these yield closed-form formulae for an optimal change of measure for importance sampling purposes.

We propose, in particular, a specific form of adaptive drift, allowing for fast computation
and increase in variance reduction.
For  geometric Asian Call options in the Heston model, MDP-based estimators with deterministic changes of drift turn out to be no better than those computed with deterministic volatility approximation in the LDP approach. 
However, MDP-based estimators with adaptive changes of drift perform much better than their LDP- counterparts with deterministic volatility approximation, 
and in fact show a performance very close to the LDP-based estimators in Heston. 
These adaptive MDP-based estimators therefore provide an efficient alternative in models where 
LDP is difficult to compute.

\textbf{Setting and notations} Throughout this paper we work on a filtered probability space $(\Omega, \mathcal{F}, \mathbb{P}, \mathbb{F})$ with a finite time horizon $T>0$, where $\Omega = \Cc([0,T]\to \RR^2)$ is the space of all continuous functions, $\mathcal F$ is the Borel-$\sigma$-algebra on $\Omega$ and $\mathbb{F}\coloneqq(\mathcal{F}_t)_{t\in[0,T]}$ is the natural filtration of a given two-dimensional standard Brownian motion~$\Wf \coloneqq (W, W^{\perp})$. 
For a pair of (possibly deterministic) process~$(X,Y)$, where $X$ is predictable and $Y$ a semi-martingale, 
we write the stochastic integral $X\circ Y := \int_{0}^{\cdot} X_s \D Y_s$ and $X\ct W := (X\circ W)_t$ for any $t\in[0,T]$. We denote any $d$-dimensional path by $\hh\coloneqq (h_1,\dots,h_d)$ for $d\in\mathbb{N}$, 
and for such a path, 
we write $\|\hh\|_T^2 := \int_{0}^{T}\left(|h_1(t)|^2+\dots + |h_d(t)|^2\right)\D t$. 
We denote $\HH^0_T$ the Cameron-Martin space of Brownian motion, isomorphic to the space of absolutely continuous functions $\mathcal{AC}([0,T])$. We define a similar space $\smash{\HH^{x}_T \coloneqq \{\varphi \in \Cc([0,T]\to \RR^d): \varphi_{t}=x+\int_{0}^{t} \psi_{s} \D s, \; \psi \in L^{2}\left([0, T] ; \RR^{d}\right)\} }$ for processes starting at $x\in\RR^d$ and a subspace $\HH_T^{x,+}\subset\HH_T^x$ where functions map to $\R^{+,d}$ instead of $\R^d$. 
Whenever a variable has an obvious time-dependence, we drop the explicit reference in the notation.
We shall also write $\Cc_T:=\Cc([0,T]\to\RR)$ to simplify statements.
We write $\{X^\eps\}\sim\LDP(\IIX, \Cc_T)$ to mean that the sequence~$\{X^\eps\}$ satisfies a large deviations principle as~$\eps$ tends to zero on~$\Cc_T$ with good rate function~$\IIX$.
For a given function~$f$, we denote by 
$\Dd(f)$ its effective domain.
We shall finally denote
$\RR^+:=(0,\infty)$.

\subsection{Overview of the importance sampling methodology}
We consider a given risk-neutral probability measure~$\PP$, so that
the fundamental theorem of asset pricing
implies that the price of an option with payoff $G\in L^2([0,T];\RR)$ is equal to $\EE^{\PP}[G]$. 
While, strictly speaking, 
we do not need $L^2([0,T];\RR)$ for pricing purposes, we require it to estimate the variance of payoff estimators.
Monte-Carlo estimators rely on the (strong) law of large number, whereby for iid samples $\{G_{i}\}_{1\leq i \leq n}$ from $\PP \circ G^{-1}$, 
the empirical mean $\widehat{G}_n :=\frac{1}{n}\sum_{i=1}^{n}G_i$
converges to the true expectation $\PP$-almost surely:
$$
\lim_{n\uparrow \infty}\widehat{G}_n = \EE^{\PP}[G],
\qquad\PP\text{-a.s.}\,.
$$
Importance sampling is a method to reduce the variance of the estimator~$\widehat{G}_n$, yielding a new law~$\QQ$ such that $\EE^{\QQ}[G] = \EE^{\PP}[G]$
and $\VV^{\QQ}[G] < \VV^{\PP}[G]$
(and of course both the equality and inequality remain true with~$G$ replaced by~$\widehat{G}_n$).
Let for example $Z \coloneqq \frac{\D\QQ}{\D\PP}$ denote the Radon-Nikodym derivative of the change of measure, so that
$\EE^{\PP}[G] = \EE^{\QQ}[G Z^{-1}]$.								 
The variance of the Monte-Carlo estimator based on iid samples of~$\widehat{G}_n Z^{-1}$ under~$\QQ$ is then
$$
\VV^{\QQ}\left[\widehat{G}_n Z^{-1}\right] 
 = \EE^{\QQ}\left[\widehat{G}_n^2 Z^{-2}\right]  - \EE^{\QQ}\left[\widehat{G}_n Z^{-1}\right]^2
 = \EE^{\PP}\left[\widehat{G}_n^2 Z^{-1}\right]  - \EE^{\PP}\left[\widehat{G}_n Z^{-1}\right]^2.
$$
If~$Z$ is chosen such that $\EE^{\PP}[\widehat{G}_n^2 Z^{-1}] < \EE^{\PP}[\widehat{G}_n^2]$, 
the variance is thus reduced.
Finding such~$Z$ however is usually hard, and  
we shall instead consider the approximation
\begin{equation}\label{eq:proxy}
\EE^{\PP}\left[\widehat{G}_n^2 Z^{-1}\right]
\approx \eps\log \EE^{\PP}\left[\exp\left\{\frac{1}{\eps}\log(G_{\eps}^2 \ Z_{\eps}^{-1})\right\}\right] 
\end{equation}
for small~$\eps>0$, for two random variables~$G_{\eps}$ and~$Z_{\eps}$ whose choices will be discussed later.
The computation of this expression is then further simplified by the use of the Varadhan's lemma (Theorem~\ref{thm:varadhan}), which casts the problem into a deterministic optimisation over the appropriate Cameron-Martin space.

\subsection{Choosing an approximated random variable~$G_{\eps}$}\label{sec:Proxies}
Consider a payoff of the form $G = G(X)$, where~$X$ is
the unique strong solution to the stochastic differential equation
\begin{equation}\label{eq:SDE}
\D X_t = b(X_t) \D t + \sigma(X_t) \D W_t, \quad X_0 = x_0,
\end{equation}
where $b,\sigma:\RR\rightarrow \RR$ are sufficiently well-behaved
deterministic functions and~$W$ is a standard Brownian motion. 
The approximation of~$G$ is then defined as  $G_{\eps}:=G(X^{\eps})$, 
where the following are possible approximations of~$X$:
\begin{definition}
Let~$X$ be the unique strong solution to~\eqref{eq:SDE}.
The process~$X^\eps$ is called
\begin{enumerate}[i)]
\item Small-noise approximation if 
\begin{equation}\label{eq:small_noise}
\D X_t^\eps = b(X_t^\eps) \D t + \sqrt{\eps} \sigma(X_t^\eps)\D W_t, \quad X_0^\eps = x_0.
\end{equation}

\item Small-time approximation if 
\begin{equation}\label{eq:small_time}
\D X_t^\eps = \eps b(X_t^\eps) \D t + \sqrt{\eps} \sigma(X_t^\eps) \D W_t, \quad X_0^\eps = x_0
\end{equation}
\item Large-time approximation if 
\begin{equation}\label{eq:large_time}
\D X_t^\eps = \frac{1}{\eps} b(X_t^\eps) \D t + \frac{1}{\sqrt{\eps}} \sigma(X_t^\eps) \D W_t, \quad X_0^\eps = x_0.
\end{equation}
\end{enumerate}
\end{definition}
The terminology here is straightforward as~\eqref{eq:small_time} follows from~\eqref{eq:SDE} via the mapping $t \mapsto \eps t$
and~\eqref{eq:large_time} follows from~\eqref{eq:SDE} via the mapping
$X_t \mapsto X_t^\eps:=X_{t/\eps}$.
The small-noise~\eqref{eq:small_noise}
comes directly from the early works 
on random perturbations of deterministic systems by Varadhan~\cite{Varadhan1967} and Freidlin-Wentzell~\cite{Freidlin2012}.


\subsection{General approach}
We consider an asset price $S\coloneqq \{S_t\}_{t\in[0,T]}$ and the corresponding log-price process $X:=\log(S) \coloneqq \{X_t\}_{t\in[0,T]}$, with dynamics
\begin{equation}\label{eq:log_price_dyn}
\begin{array}{rll}
\def\arraystretch{1.5}
\D X_t & = \displaystyle -\frac{1}{2}V_t \D t + \sqrt{V_t} \D B_t, & X_0 = 0, \\
\D V_t & = f(V_t)\D t + g(V_t)\D W_t, & V_0 = v_0>0,
\end{array}
\end{equation}
where $\Wf = (W,\Wp)$ is a standard two-dimensional Brownian motion and $B\coloneqq \rho W + \rrho\Wp$ with correlation coefficient $\rho \in (-1,1)$ and $\rrho\coloneqq \sqrt{1-\rho^2}$. 
The drift and diffusion coefficients of the volatility process satisfy  $f:\RR^+\rightarrow\RR$ and $g:\RR^+\rightarrow\RR^+$ and Assumption~\ref{ass:SDE_coeffs} if not stated otherwise (e.g. in the case of large-time approximation in Section~\ref{sec:largetimeMDP} additional assumptions are required for ergodicity purposes). 
\begin{assumption}\label{ass:SDE_coeffs}\
\begin{enumerate}[i)]
\item The function $f:\RR^+\rightarrow\RR$ is globally Lipschitz continuous;
\item The function $g:\RR^+\rightarrow (0,\infty)$ is strictly increasing 
and satisfies $p$-polynomial growth condition for $p\geq\frac{1}{2}$, that is 
$|g(x)|\leq 1 + |x|^p$,
for all $x\in\RR^+$.
\end{enumerate}
\end{assumption}
Under this assumption,
Yamada-Watanabe's theorem~\cite[Theorem 1]{Yamada1971} ensures the existence of 
a unique strong solution to~\eqref{eq:log_price_dyn}.
Consider now the continuous map
$G\in\Cc_T$, 
yielding the option price~$\EE[G(X)]$.
Finding the optimal Radon-Nikodym derivative encoding the change of measure from~$\PP$ to~$\QQ$ is hard in general and we instead consider the particular class of change of measure
\begin{equation}\label{eq:ZRadon}
Z_T\coloneqq \left.\frac{\D \QQ}{\D \PP}\right\vert_{\Ff_T} = \exp\left\{-\frac{1}{2}\|\dot{\hh}\|_{T}^2 + \dot{\hh}\ct \Wf^\top \right\},
\end{equation}
which is well defined and satisfies $\EE^{\PP}[Z_T]=1$
whenever $\hh \in \HH^0_T$.
Now let $F := \log|G|$ and $H := \log(Z)$, so that
$\EE^{\PP}[Z_T^{-1} G(X)^2] = \EE^{\PP}[\E^{2F(X) - H_T}]$; 
the approximation~\eqref{eq:proxy} then yields the estimate
\begin{equation}\label{eq:MinProblem}
L(\hh):=\limsup_{\eps\downarrow 0}\eps\log\EE^{\PP}\left[\exp\left\{\frac{2F(X^\eps) - H_T^\eps(\hh)}{\eps}\right\}\right],
\end{equation}
to compute, for some proxies~$X^\eps$ and~$H_T^\eps$.
In light of~\eqref{eq:ZRadon}, the approximation~$H^\eps(\hh)$  reads
\begin{equation}\label{eq:ProxyHEps}
H^{\eps}_T(\hh) = -\frac{1}{2}\|\dot{\hh}\|^2 + \dot{\hh}\circ_T \Wf^{\eps\top},
\end{equation}
with $\Wf^\eps:=\sqrt{\eps}\,\Wf$.
\begin{definition}\label{def:asyoptdrift}
A path $\hh\in \HH_T^0$ 
minimising~$L(\hh)$
is called an asymptotically optimal change of drift.
\end{definition}

From this point onward, several approaches exist in the literature. In~\cite{Dupuis2012, Hartmann2015, Dupuis2017} fully adaptive schemes are consider, where $\hh$ is function of $(t, X_t, V_t)$. These schemes effectively reduce variance, but are expensive to compute. 
For that reason, we look at the case where~$\hh$ is an absolutely continuous function with derivatives in $L^2([0,T]; \RR)$.
The main advantage is the fast computation in comparison to the fully adaptive schemes. 
Conserving this advantage, we also look at path~$\hh$ of the form $\int_0^.\dot{\hh}(t)\sqrt{V_t}\D t$ (yielding a stochastic change of measure) for which computations are usually as fast and variance reduction higher. 

In the case where~$\hh$ is a deterministic function 
(the approach is similar in other settings), 
the main methodology we shall develop below then goes through the following steps:

\begin{enumerate}[i)]
    \item Choose appropriate approximations~$X^{\eps}$ and~$H^{\eps}$ as in Section~\ref{sec:Proxies};
    \item Prove an LDP (MDP) with good rate function~$\II^{X,H}$ for $\{X^\eps,H^{\eps}\}_{\eps>0}$;
    \item Show that Varadhan's Lemma applies, 
    so that the function~$L(\hh)$ in~\eqref{eq:MinProblem} reads
    \begin{equation}\label{eq:Primal_L}
    L(\hh) = \sup_{(X,H)\in \Dd(\II^{X,H})} \Lf(X,H;\hh),
    \end{equation}
    where
\begin{equation}\label{eq:TargetFunc}
\Lf(X,H;\hh) := 2F(X) - H_T(\hh) - \II^{X,H}(X,H).
    \end{equation}
    \item We consider the dual problem of \eqref{eq:Primal_L} in the sense of Definition~\ref{def:Dual_L}. For more details see the remark below.
\end{enumerate}

\begin{definition}\label{def:Dual_L}
The primal problem is defined as
\begin{equation}\label{eq:Primal}
\inf_{\hh\in\HH^T_0} \sup_{(X,H)\in \Dd(\II^{X,H})}\Lf(X,H;\hh),
\end{equation}
while the dual consists in
\begin{equation}\label{eq:Dual}
\sup_{(X,H)\in \Dd(\II^{X,H})} \inf_{\hh\in\HH^T_0} \Lf(X,H;\hh).
\end{equation}
\end{definition}

\begin{remark}
In many cases this optimisation problem may be difficult to solve analytically, so we deal with the dual problem which turns out to be much simpler.
With further assumptions, it may be possible to prove strong duality, however this is outside the scope of this paper.
\end{remark}

\begin{remark}
Small-time approximations may induce important loss of information. 
The reason is that the drift term in~$H^\eps$ (the quadratic part in~$\dot{h}$) can be negligible and can thus lead to a trivial dual problem. 
In the Black-Scholes setting (Appendix~\ref{sec:BS}),
a small-time approximation for~$H$ leads to the following problem: 
let $F:\RR^+\rightarrow\RR^+$ be a smooth enough function so that Varadhan's Lemma (Theorem~\ref{thm:varadhan}) holds, then the small-noise approximation problem reads
$$
\Lf(\xm,H;\hm) = 2 F(x)  - \int_0^T\dot{\hm}_t \dot{\xm}_t\D t
- \half\int_0^T\dot{\xm}_t^2\D t
+ \half\int_0^T\dot{\hm}_t^2\D t.
$$
However, in the small-time case we actually have
$$
\Lf(\xm,H;\hm) = 
2F(\xm) - \int_0^T\dot{\hm}_t \dot{\xm}_t\D t
- \half\int_0^T\dot{\xm}_t^2\D t.
$$
In this small-time setting, the dual problem~\eqref{eq:Dual} then reads
\begin{align*}
\sup_{(\xm,H)\in\Dd(\II^{X,H})}\inf_{\hm\in\HH^T_0} \Lf(X,H;\hm)
 & = \sup_{(\xm,H)\in \Dd(\II^{X,H})} \inf_{\hm\in\HH^T_0} \Big\{2F(\xm) - \int_0^T\dot{\hm}_t \dot{\xm}_t\D t - \half\int_0^T\dot{\xm}_t^2\D t\Big\}\\
 & = \sup_{(\xm,H)\in \Dd(\II^{X,H})} \Big\{2F(\xm)
 - \half\int_0^T\dot{\xm}_t^2\D t
 - \sup_{\hm\in\HH^T_0}\int_0^T\dot{\hm}_t \dot{\xm}_t\D t\Big\}.
\end{align*}
Clearly the path~$\hm$ can be multiplied by any arbitrarily large positive constant to increase the inner supremum, 
and therefore the optimisation does not admit a maximiser.
In these cases, we thus do not consider the small-time approximation for~$H^{\eps}_T$.
\end{remark}



The paper's structure is as follows: 
in Section~\ref{sec:LDP_IS} we look at stochastic volatility models satisfying Assumption~\ref{ass:SDE_coeffs} and derive explicit solutions for large deviations approximations for \textit{path-dependent} payoffs of the form $F\big(\smallint _0^T\alpha_t\D X_t\big)$ for general deterministic  paths $\alpha\in\Cc_T$. This includes \textit{state-dependent} payoffs of European type, i.e., functions of $X_T$ (for the choice of $\alpha = 1$) and of Asian type $\frac{1}{T}\smallint_0^TX_t\D t$ (for $\alpha_t = \frac{T-t}{T}$). Later in Section~\ref{sec:MDP_IS} we study moderate deviations, where we derive small-noise, small-time and large-time MDPs, whose advantages, compared to LDP, are simpler forms of rate functions. Finally, in Section~\ref{sec:Num_results} we present results for the Heston model and compare variance reduction results for different approximation types. Some of the technical proofs are relegated to the Appendix~\ref{apx:proofs}.

\section{Importance sampling via large deviations}\label{sec:LDP_IS}


\subsection{Small-noise LDP}
We start with the small-noise approximation of~\eqref{eq:log_price_dyn}:
\begin{equation}\label{eq:log_price_dynNoise}
\left\{
\begin{array}{rll}
\D X_t^\eps & = \displaystyle  -\frac{1}{2}V_t^\eps \D t + \seps\sqrt{V_t^\eps} \D B_t, & X_0^\eps = 0,\\
\D V_t^\eps & = \displaystyle   f(V_t)\D t + \seps g(V_t)\D W_t, & V_0^\eps = v_0.
\end{array}
\right.
\end{equation}


\subsubsection{Large deviations}
In the spirit of~\cite{Freidlin2012}, 
we provide in this section an LDP in ${\Cc([0,T]\to\RR^2)}$ for $\{X^\eps, V^\eps\}$;
usual assumptions involve non-degenerate and locally-Lipschitz diffusion though, 
which clearly fails for square-root type stochastic volatility models. 
We follow~\cite{Robertson2010} who lifted this constraint and instead showed an LDP under the following assumption:
\begin{assumption}\label{ass:LDP_robertson}
$\{V^{\eps}\}_{\eps}\sim\LDP(\IIV, \Cc_T)$
for some good rate function~$\IIV$ 
with $\Dd(\IIV)\subset \HH_{T}^{v_0}$.
\end{assumption}
This assumption allows an extension of the Freidlin-Wentzell in the following theorem, which is a log-price analogue of \cite[Theorem 2.2]{Robertson2010} and also gives two variational conditions sufficient for the tuple to satisfy an LDP.
\begin{theorem}[Theorem 2.2 and Corollary 2.3 in~\cite{Robertson2010}]\label{thm:robertsonLDP}
Under Assumption~\ref{ass:LDP_robertson}, 
if there exists $\beta>0$ such that
\begin{equation}\label{eq:robertsonLDPcond}
\gamma \coloneqq \sup_{\phi \in \Dd(\IIV)}\left\{\beta \int_{0}^{T} \frac{\dot{\phi}_{t}^{2}}{4\phi_{t}}  \mathrm{~d} t - \IIV(\phi)\right\}
\end{equation}
is finite, then 
$\{X^{\eps}, V^{\eps}\}\sim\LDP\left(\IIXV,\Cc\left([0,T]\to\RR^2\right)\right)$, where
$$
\IIXV(\phi, \psi)=\inf \left\{\IIV(\phi) + \IIW(w): \psi :=-\frac12\int_{0}^{\cdot}\phi_t\D t+\int_{0}^{\cdot} \sqrt{\phi_{t}} \dot{w}_{t} \D t\right\},
$$
with
$\IIW(w) = \frac{1}{2} \int_{0}^{T} \dot{w}_{t}^{2} \D t$
if $w \in \HH_{T}^{0}$ and infinite otherwise.
Furthermore, $\{X^{\eps}\}\sim\LDP(\IIX, \Cc_T)$ with
$\IIX(\psi):=\inf \left\{\IIXV(\phi, \psi): \phi \in \Cc_T\right\}$.
\end{theorem}
Here, $\IIW$ is nothing else than the usual energy function for the Brownian motion from Schilder's theorem~\cite{Schilder1966}.
To apply Theorem~\ref{thm:robertsonLDP}, we first need to show that Assumption~\ref{ass:LDP_robertson} holds in our setting and check whether any further assumption on the coefficients are necessary. 
Many processes arising from volatility models, where the classical Freidlin-Wentzell does not apply, have been studied in the literature. 
For example Donati-Martin, Rouault, Yor and Zani~\cite{Donati-Martin2004} show that~$V^{\eps}$ satisfies an LDP in the case of the Heston model, then Chiarini and Fischer~\cite{Chiarini2014} show existence of an LDP for class of models with uniformly continuous coefficients on compacts, Conforti et al.~\cite{conforti2015small} show an LDP for non-Lipschitz diffusion coefficient of CEV type. 
Most notably, Baldi and Caramellino~\cite{Baldi2011} cover the case of Assumption~\ref{ass:SDE_coeffs} with $f(0)>0$ and sub-linear growth when $f,g\rightarrow \infty$. We now state their main result.
\begin{theorem}[Theorem 1.2 in~\cite{Baldi2011}]
Let $V^{\eps}$ be the solution of~\eqref{eq:log_price_dynNoise} on $[0, T]$ with $v_0>0$, $f(0)>0$ and sublinear growth of $f,g$ at infinity. 
Then under Assumption~\ref{ass:SDE_coeffs} the process~$V^\eps$ satisfies an LDP with the good rate function
\[
\IIV(\phi)= \begin{cases}
\displaystyle \frac{1}{2} \int_{0}^{T} \dot{w}_t^2 \D t, 
& \text{if } w\in\HH^0_T, \dot{\phi}_t = f(\phi_t) + g(\phi_t)\dot{w}_t, \phi_0 = v_0, \\ 
+\infty, & \text{otherwise}.
\end{cases}
\]
\end{theorem}
Since $\seps \Wp$ is independent of $V^{\eps}, \seps W$ and satisfies a LDP with good rate function by Schilder's theorem, 
then~\cite[Exercise 4.2.7]{Dembo2010} in turn implies that the triple $(V^{\eps}, \Wf^\eps)$ satisfies an LDP with the good rate function
$$
\II^{V,\Wf}(\phi,\ww) =  
\frac{1}{2}\|\dot{\ww}\|_{T}^2,
\quad\text{ if }\quad
\ww \coloneqq \begin{pmatrix}
           w \\
           \widetilde{w}
         \end{pmatrix} \in \HH_T^0\times \HH_T^0,\; \dot{\phi}_t = f(\phi_t) + g(\phi_t)\dot{w}_t\; \text{ and } \; \phi_0 = v_0,
$$
and is infinite otherwise. 
Then, by applying Theorem~\ref{thm:robertsonLDP}, we finally obtain the following:
\begin{proposition} 
Let $\Yf^{\eps} \coloneqq \sqrt{V^\eps}\circ \Wf^\eps$.
The triple $\{V^{\eps}, \Yf^{\eps}\}$ satisfies an LDP with good rate function
\begin{equation*}
\II^{V,\Wf,\Yf}(\phi,\ww,\yy) = \frac{1}{2} \|\dot{\ww}\|_T^2,
\quad \text{if} \quad 
\left\{
 \begin{array}{l}
\displaystyle \ww \in \HH_T^0\times \HH_T^0,\\
\displaystyle \dot{\phi}_t = f(\phi_t) + g(\phi_{t})\dot{w}_t \; \text{ and } \;  \phi_0 = v_0, \\
\displaystyle \dot{\yy} = \sqrt{\phi}\dot{\ww},
 \end{array}
\right.
\end{equation*}
and is infinite otherwise.
\end{proposition}
\begin{proof}
The proposition is a direct application of Theorem~\ref{thm:robertsonLDP}:
since $\Yf^{\eps} = \sqrt{V^{\eps}}\circ\Wf^\eps$ is a continuous function of $V^\eps$, the result follows from the Contraction principle~\cite[Theorem~4.2.1]{Dembo2010}. 
It remains to verify that condition~\eqref{eq:robertsonLDPcond} holds:
\begin{align*}
\beta\int_0^T\frac{\dot{\phi}_t}{4\phi_t}\D t - \IIV(\phi) &= \beta\int_0^T\frac{\dot{\phi}_t}{4\phi_t}\D t - \frac{1}{2}\int_0^T\left|\frac{\dot{\phi}_t-f(\phi_t)}{g(\phi)} \right|^2 \D t \\ &= \frac{1}{2}\int_0^T \left( \frac{\beta}{2\phi_t} - \frac{1}{g^2(\phi_t)} \right ) \dot{\phi}_t^2 \D t + \int_0^T \frac{f(\phi_t)}{g^2(\phi_t)}\dot{\phi}_t\D t - \frac{1}{2}\int_0^T \left | \frac{f(\phi_t)}{g(\phi_t)} \right |^2 \D t.
\end{align*}
By~\cite[Proposition 3.11]{Baldi2011} the unique solution to the ODE $\dot{\phi_t}=f(\phi_t)+g(\phi_t)\dot{w}_t$ with $\phi_0 = v_0 > 0$ is strictly positive under Assumption~\ref{ass:SDE_coeffs} and since $g$ maps to $\R^+$ and is strictly increasing, 
then $\sup_{t\in[0,T]} \frac{2\phi_t}{g^2(\phi_t)}$ is finite.
Again, since~$\phi$ is strictly positive, then $\inf_{t\in[0,T]}\frac{2\phi_t}{g^2(\phi_t)}>0$ 
and therefore for all $0<\beta\leq\inf_{t\in[0,T]}\frac{2\phi_t}{g^2(\phi_t)}$,
$$
\beta\int_0^T\frac{\dot{\phi}_t}{4\phi_t}\D t - \IIV(\phi) 
\leq \int_0^T \frac{f(\phi_t)}{g^2(\phi_t)}\dot{\phi}_t \D t 
\leq \sup_{t\in[0,T]}\left |  \frac{f(\phi_t)}{g^2(\phi_t)} \right | \int_0^T \dot{\phi}_t \D t
< \infty.
$$
\end{proof}

\subsubsection{LDP-based importance sampling}
We consider two changes of measure, with a deterministic and a stochastic change of drift,
and start with the former.
\paragraph{\textit{Deterministic change of drift}}
The drift is of the form
$$ 
\left.\frac{\D\QQ}{\D\PP}\right|_{\Ff_T} := \exp\left\{\dot{\hh}\cT \Wf^\top - \frac{1}{2}\|\dot{\hh}\|_{T}^2\right\},
$$
with $\hh\in \HH_T^0$.
The limit~\eqref{eq:MinProblem}, together with~\eqref{eq:ProxyHEps}, then reads
$$
L(\hh) = \limsup_{\eps\downarrow 0}
\eps\log\EE\left[
\exp\left\{\frac{1}{\eps}\left(2F(X^{\eps}) - \dot{\hh}\cT \Wf^{\eps\top} + \frac{1}{2}\|\dot{\hh}\|_{T}^2\right)\right\}
\right] .
$$
We now follow the same approach as in the case of deterministic volatility in Appendix~\ref{sec:BS}. 
Let $F\in\Cc(\RR^+\rightarrow\RR^+)$ be bounded from above and~$\dot{\hh}$ be of finite variation, then the tail condition of Varadhan's Lemma is easy to verify and 
the functional~$L$ in~\eqref{eq:Primal_L} reads
$$
L(\hh) = \sup\limits_{\xx\in \HH_T^0}
\Lf(\xx,\hh),
\qquad\text{with}\qquad
\Lf(\xx,\hh) := 2 F(\varphi_T(\xx)) - \dot{\hh}\cT \xx^\top
 + \half\left(\|\dot{\hh}\|_{T}^2 - \|\dot{\xx}\|_{T}^2\right),
 $$
where $\varphi(\xx)$ is the unique solution on $[0,T]$ to
$$
\dot{\varphi}_t(\xx) = -\frac{1}{2}\psi_t + \sqrt{\psi_t}\vrho\dot{\xx}(t)^\top,
\qquad\text{with}\qquad
\dot{\psi}_t = \kappa (\theta - \psi_t) + \xi \sqrt{\psi_t} \dot{x}_1(t),
$$
with initial conditions $\varphi_0(\xx) = 0$ and $\psi_0 = v_0$, and $\vrho:=(\rho, \brho)$. 
To solve the dual problem~\eqref{eq:Dual}, the inner optimisation reads
\begin{align*}
\inf_{\hh\in\HH_T^0}\Lf(\xx,\hh) &= \inf_{\hh\in\HH_T^0} \left\{2F(\varphi_T(\xx)) - \dot{\hh}\cT \xx^\top
 + \half\left(\|\dot{\hh}\|_{T}^2 - \|\dot{\xx}\|_{T}^2\right)\right\} \\
&= 2F(\varphi_T(\xx)) + \inf_{\hh\in\HH_T^0}\left\{\frac{\|\dot{\hh}\|_{T}^2}{2} -\dot{\hh}\cT \xx^\top - \frac{\|\dot{\xx}\|_{T}^2}{2}\right\} \\
&= 2F(\varphi_T(\xx)) - \|\dot{\xx}\|_{T}^2.
\end{align*}
As shown in~\cite{Robertson2010}, the dual problem~\eqref{eq:Dual}
can then be solved uniquely as 
\begin{equation}\label{eq:detoptdrift_dual}
\hh^* = \argmax\limits_{\xx\in \HH_T^0} 
\left\{F(\varphi_T(\xx)) - \frac{\|\dot{\xx}\|_{T}^2}{2}\right\},
\end{equation}
which is an asymptotically optimal change of drift in the sense of Definition~\ref{def:asyoptdrift}.
\paragraph{\textit{Stochastic change of drift}}
We now consider the stochastic change of measure 
$$
\left.\frac{\D\QQ}{\D\PP}\right|_{\Ff_T} := \exp\left\{\left(\dot{\hh}\sqrt{V}\right)\cT \Wf^\top - \frac{1}{2}\left\|\dot{\hh}\sqrt{V}\right\|_{T}^2\right\},
$$
with~$\dot{\hh}$ a deterministic function of finite variation such that $\EE[\frac{\D\QQ}{\D \PP}]=1$;
this holds for example under the Novikov condition $\EE\left [\exp(\frac{1}{2}\int_0^T \|\dot{\hh}(t)\|^2 V_t \D t)\right ]<\infty$. 
We again consider $G\in\Cc_T$ and let $F:=\log|G|$. 
The minimisation problem~\eqref{eq:MinProblem} now reads
\begin{align*}
L(\hh) &=\limsup\limits_{\eps\downarrow 0}\eps
\log\EE\left[\exp\left\{\frac{1}{\eps}\left(2F(X^{\eps}) - (\dot{\hh}\sqrt{V^{\eps}}) \cT \Wf^{\eps\top} + \frac{1}{2}\|\dot{\hh}\sqrt{V^\eps}\|_{T}^2\right)\right\}\right]\\
&= \limsup\limits_{\eps\downarrow 0}\eps\log\EE\left[\exp\left\{\frac{1}{\eps}\left(2F(X^{\eps}) - \dot{\hh} \cT \Yf^{\eps\top} + \frac{1}{2}\|\dot{\hh}\sqrt{V^\eps}\|_{T}^2\right)\right\}\right].
\end{align*}
Since $F$ is continuous, the term inside the exponential is a continuous function of $\{V^{\eps},\Yf^{\eps}\}$. 
Varadhan's Lemma then yields
$L(\hh) = \sup_{\xx\in \HH_T^0}\Lf(\xx,\hh)$,
with 
$$
\Lf(\xx,\hh) = 2F(\varphi_T(\xx)) - \left(\dot{\hh}\sqrt{\psi}\right)\cT \xx^\top + \frac{\|\dot{\hh}\sqrt{\psi}\|_{T}^2}{2} - \frac{\|\dot{\hh}\|_{T}^2}{2},
$$
where $\xx = (x_1 \; x_2)^\top$ and $\{\varphi(\xx),\psi\}$ are unique solutions on $[0,T]$ to
$$
\dot{\varphi}_t(\xx) = -\frac{1}{2}\psi_t + \sqrt{\psi_t} \vrho\dot{\xx}(t)^\top,
\qquad\text{with}\qquad
\dot{\psi}_t(\xx) = f(\psi_t) + g(\psi_t) \dot{x}_1(t),
$$
with initial conditions $(\varphi_0(\xx),\psi_0) = (0,v_0)$. 
For  the dual problem, we search for a change of measure with $\hh$ such that:
\begin{equation}\label{eq:detoptstochdrift_dual}
\left\{
\begin{array}{rll}
\xx^* & = \displaystyle \argmax\limits_{\xx\in \HH_T^0} 
\left\{F(\varphi_T(\xx)) - \frac{\|\dot{\xx}\|_{T}^2}{2}\right\},\\
\dot{\varphi}_t & = \displaystyle -\frac{1}{2}\psi_t + \sqrt{\psi_t} \vrho \dot{\xx}(t)^\top, & \varphi_0 = 0,\\
\dot{\psi}_t & = \displaystyle f(\psi_t) + g(\psi_t) \dot{x}_1(t), & \psi_0 = v_0,\\
\hh^* & = \displaystyle \int_{0}^{\cdot}\frac{\dot{\xx}^*(t)}{\sqrt{\psi_t}}\D t.
\end{array}
\right.
\end{equation}
The maximisation problem is very similar to the one with deterministic change of drift~\eqref{eq:detoptdrift_dual}. 
However, as we will see in Section~\ref{sec:Num_results}, the stochastic version usually gives better results.

\subsubsection{Example: options with path-dependent payoff}
\label{sec:PayoffPathDepG}
 Consider a payoff $G(\alpha\cT X)$ with ${G:\RR^{+} \rightarrow \RR^{+}}$ differentiable, 
$\alpha$ a positive function of class~$\Cc^1$ and $F := \log|G|$. 
We only look at the deterministic case, 
namely the optimisation problem~\eqref{eq:detoptdrift_dual} since the solutions to~\eqref{eq:detoptstochdrift_dual} can be easily deduced from it. It reads
\begin{equation*}
\left\{
\begin{array}{rll}
\hh^* & = \displaystyle \argmax\limits_{\xx\in \HH_T^0} 
\left\{F\left(\int_0^T\alpha_t \dot{\varphi}_t \D t\right) - \frac{\|\dot{\xx}\|_{T}^2}{2}\right\},\\
\dot{\varphi}_t & = \displaystyle -\frac{1}{2}\psi_t + \sqrt{\psi_t}\vrho \dot{\xx}(t)^\top, & \varphi_0 = 0,\\
\dot{\psi}_t & = \displaystyle f(\psi_t) + g(\psi_t) \dot{x}_1(t), & \psi_0 = v_0.
\end{array}
\right.
\end{equation*}
The following lemma helps transforming the above optimisation problem.
\begin{lemma}\label{lem:Simplify}
Let $\xx\in \HH_T^0 \times \HH_T^0$. The function $K:\HH_T^0 \times \HH_T^0 \rightarrow \HH_T^0\times \HH_T^{v_0,+}$ such that $K(\xx)=(\varphi,\psi)$ is solution to
$$
\dot{\varphi}_t  = -\frac{1}{2}\psi_t + \sqrt{\psi_t}\vrho \dot{\xx}(t)^\top,
\qquad\text{with}\qquad
\dot{\psi}_t = f(\psi_t) + g(\psi_t) \dot{x}_1(t),
$$
with initial conditions $\varphi_0 = 0$ and $\psi_0 = v_0$, is well defined and is a bijection.
\end{lemma}
\begin{proof}
Clearly $\dot{\psi}\in L^{2}\left([0, T]; \RR\right)$ and the unique solution $\psi$ to the ODE $\dot{\psi_t}=f(\phi_t)+g(\psi_t)\dot{w}_t$ with $\psi_0 = v_0 > 0$ is strictly positive under Assumption~\ref{ass:SDE_coeffs} by~\cite[Proposition 3.11]{Baldi2011}. Therefore $\psi\in\HH_T^{v_0, +}$ and $K(\xx)=(\varphi, \psi)$ is well defined.
Finally, $K$ is clearly a bijection and its inverse can be computed explicitly as
\[
K^{-1}(\mathbf z) = \left(\int_0^t \frac{\dot{z}_1(t)-f(z_2(t))}{g(z_2(t))}\D t,\; \frac1\brho\int_0^t\left\{ \frac{\dot{z}_1(t) + \frac{1}{2}z_2(t)}{\sqrt{z_2(t)}} - \rho\frac{\dot{z}_2(t)-f(z_2(t))}{g(z_2(t))} \right\} \D t\right).
\]
\end{proof}
Using Lemma~\ref{lem:Simplify}, we can substantially simplify the optimisation problem by writing it in terms of $K(\xx)$. To be more precise, we will make use of the following transformation:
\begin{equation}\label{eq:transformation_UZ}
U(\psi) = \frac{\dot{\psi} - f(\psi)}{g(\psi)}
\qquad\text{and}\qquad
Z(\varphi,\psi) = \frac{\dot{\varphi} + \frac{1}{2}\psi}{\sqrt{\psi}},
\end{equation}
which stems from the two components of $K(\xx)$. The optimisation problem becomes
\begin{equation*}
\left\{
\begin{array}{rl}
\varphi^{*},\psi^{*} & = \displaystyle\argmax\limits_{\{\varphi,\psi\} \in \HH_T^0\times \HH_T^{v_0,+}} \left\{F\left(\int_0^T\alpha_t \dot{\varphi}_t \D t\right) - \frac{1}{2}\int_0^T \left\{U(\psi_t)^2 + \frac{1}{\brho^2}(Z(\varphi_t,\psi_t) - \rho U(\psi_t))^2\right\}\D t\right\},\\
\dot{h}_1^{*} & = U(\psi^{*}), \\
\dot{h}_2^{*} & = \displaystyle \frac{1}{\brho}(Z(\varphi^{*},\psi^{*}) - \rho U(\psi^{*})).
\end{array}
\right.
\end{equation*}
This allows us to apply Euler-Lagrange to the problem seen as an optimisation over $\left\{\int_{0}^{\cdot} \alpha_t \dot{\varphi}_t \D t,\,\psi\right\}$:
\begin{align*}
\frac{\D}{\D t}\left\{ \frac{2}{\brho^2} \frac{(Z - \rho U)}{\alpha \sqrt{\psi}} \right\} &= 0, \\
\frac{\D}{\D t}\left\{ \frac{2U}{g(\psi)} - \frac{2\rho}{\brho^2}\frac{(Z - \rho U)}{g(\psi)} \right\} &= - \frac{2U}{g^2(\psi)}\zeta + \frac{2}{\brho^2}(Z-\rho U)\left\{ \frac{1}{2\sqrt{\psi}} - \frac{Z}{2\psi} + \frac{\rho}{g^2(\psi)}\zeta \right\}, 
\end{align*}
where $\zeta \coloneqq f'(\psi)g(\psi) + (\dot{\psi}-f(\psi))g'(\psi)$. This system of equations is still hard to solve for general $f$ and $g$, so we instead consider the case of the Heston model with $f(\psi)=\kappa(\theta - \psi)$ and $g(\psi) = \xi \sqrt{\psi}$.
Then the system becomes
\begin{align*}
\frac{\D}{\D t}\left\{\frac{2}{\brho^2} \frac{(Z - \rho U)}{\alpha\sqrt{\psi}}\right\} &= 0,\\
\frac{\D}{\D t}\left\{2\frac{U}{\xi \sqrt{\psi}} -2\frac{\rho}{\brho^2}\frac{(Z - \rho U)}{\xi\sqrt{\psi}} \right\} &= 2U\left(-\frac{U}{2\psi} + \frac{\kappa}{\xi \sqrt{\psi}}\right)
+ \frac{2}{{\Bar\rho}^2}(Z - \rho U)\left(-\frac{Z-\rho U}{2\psi} + \frac{\frac{1}{2} - \rho \frac{\kappa}{\xi}}{\sqrt{\psi}}\right).
\end{align*}
Solving the first ODE and plugging it into the second one gives for all $t\in[0,T]$
\begin{align*}
\frac{1}{\brho^2} (Z_t - \rho U_t) &= \frac{1}{2} \beta \alpha_t \psi_t,\\
\dot{A}_t &= -\frac{1}{2}\xi A_t^2 + \kappa A_t + \frac{1}{2}\xi \beta \alpha_t \left(\frac{1}{2} - \frac{1}{4}\brho^2\beta\alpha_t - \rho\frac{\kappa}{\xi} \right) + \frac{1}{2}\rho\beta\dot{\alpha}_t,
\end{align*}
where $\beta>0$ is an arbitrary constant and $A_t = \frac{U_t}{\sqrt{\psi_t}}$ is the solution to the Ricatti equation. The option payoff at the terminal time determines the boundary conditions through the Euler-Lagrange equations so that
\[
\beta = - \frac{2F_T'}{\sqrt{\psi_T}} 
\qquad \text{and} \qquad 
A_T = -\frac{\rho}{\xi}F_T'\alpha_T,
\]
with $F_T=F(\smallint_0^T\alpha_t\dot{\varphi}_t \D t)$. 
Since both conditions include the same optimising variable, 
the resulting problem in fact becomes an optimisation in~$\RR^2$ over~$\beta$ and~$A_T$ (or equivalently~$A_0$) and is therefore much simpler than the original optimisation problem. 
The procedure is the following: after solving for $A_t$ for all $t\in[0,T]$, we can now solve for the couple $(\varphi,\psi)$ by writing $U$ and $Z$ in terms of the couple using \eqref{eq:transformation_UZ}. 
We note that in the small-noise setting the results are of the similar form.
\begin{example}
In the case $\alpha_t=\alpha>0$, when the Riccati equation can be reduced to a separable differential equation, 
let $C(\beta)\coloneqq \frac12\xi\beta\alpha\left(\frac12-\frac14\brho^2\beta\alpha-\rho\frac\kappa\xi\right)$ so that the Riccati equation reads
\[
\frac{\D A_t}{\D t} = -\frac\xi 2 A^2_t + \kappa A_t + C(\beta),
\]
the solution to which is\[
A_t = \frac1\xi\left\{\sqrt{2\xi C(\beta)-\kappa^2}\tan\left(-\frac12\frac{t-D(A_0, \beta)}{\sqrt{2\xi C(\beta)-\kappa^2}}\right)+\kappa\right\},
\]
where $D\in\RR$ can be determined from the initial condition on $A$. Now, since $A_t=\tfrac{U_t}{\sqrt{\psi_t}}=\tfrac{\dot\psi_t-\kappa(\theta-\psi_t)}{\psi_t}$, we are to solve
$\dot\psi_t + (\kappa - A_t)\psi_t = \kappa\theta$,
for $t \in[0,T]$,
which is just a non-homogeneous linear ODE with the solution
$$
\psi_t = \left(\kappa\theta\int_0^t 
\exp\left\{\int_0^s (\kappa-A_u)\D u\right\}\D s + v_0\right)
\exp\left\{-\int_0^t (\kappa-A_s)\D s\right\}.
$$
Then the optimisation problem for $\hh$ reduces to
\begin{equation*}
\left\{
\begin{array}{rcll}
h_1(t)
 & = & \displaystyle U(\psi^{*})
 & = \displaystyle \min_{A_0, \beta>0} \frac{\dot \psi_t(A_0, \beta)-\kappa(\theta-\psi_t(A_0, \beta))}{\psi_t(A_0, \beta)} \\
h_2(t)
& = & \displaystyle \frac{1}{\brho} \Big(Z(\varphi^{*}, \psi^{*})-\rho U(\psi^{*})\Big)
& = \displaystyle \frac{\alpha}{2}\min_{A_0, \beta>0} \beta\psi_t(A_0, \beta),
\end{array}
\right.
\end{equation*}
which is an optimisation over $(A_0,\beta)\in \RR^2$.
\end{example}
\subsection{Small-time LDP}
Applying the mapping $t\mapsto \eps t$ to~\eqref{eq:log_price_dyn} yields
\begin{equation}\label{eq:log_price_dynSmallTime}
\begin{array}{rll}
\D X_t^\eps & = \displaystyle  -\frac{1}{2}\eps V_t^\eps \D t + \seps\sqrt{V_t^\eps} \D B_t,
 & X_0^\eps = 0,\\
\D V_t^\eps & = \displaystyle \eps f(V_t^\eps)\D t + g(V_t^\eps) \D W^{\eps}_t, & V_0 = v_0.
\end{array}
\end{equation}
Robertson~\cite{Robertson2010} showed that $\eps\int_{0}^{\cdot} V_t^\eps \D t$ is in fact exponentially equivalent to zero,
so that the drift of~$V^\eps$ can be ignored at the large deviations level. 
In the case of a general drift~$f$, the following lemma provides a similar statement:
\begin{lemma}
The process $\eps\int_{0}^{\cdot} f(V_t^\eps) \D t$ is exponentially equivalent to zero.
\end{lemma}
\begin{proof}
Markov's inequality implies that, for any $\delta>0$,
\[
\PP\left[ \int_0^T f(V_s^\eps)\D s > \frac{2\delta}{\eps} \right] \leq \exp\left\{ -\frac{2\delta}{\eps^2}\right\}\EE\left[ \exp\left\{ \frac 1\eps \int_0^T f(V_s^\eps)\D s \right\} \right],
\]
and we are therefore left to show that 
$\limsup_{\eps\downarrow 0}\eps\log\EE\left[ \exp\left\{ \frac 1\eps \int_0^T f(V_s^\eps)\D s \right\} \right]$
is finite.
To that end we apply the integral Jensen's inequality
\[
\EE\left[ \exp\left\{ \frac 1\eps \int_0^T f(V_s^\eps)\D s \right\} \right] \leq \frac{1}{T} \int_0^T \EE\left[ \exp\left\{ \frac{T}\eps f(V_s^\eps) \right\} \right] \D s
\]
and the linear growth condition from the global Lipschitz condition in Assumption~\ref{ass:SDE_coeffs}:
\begin{align*}
\frac{1}{T} \int_0^T \EE\left[ \exp\left\{ \frac{T}\eps f(V_s^\eps) \right\} \right] \D s &\leq \frac{1}{T} \int_0^T \EE\left[ \exp\left\{ \frac{T}\eps\left(1 + |V_s^\eps| \right) \right\} \right] \D s 
= \frac{1}{T}\E^{\frac{T}{\eps}} \int_0^T \EE\left[ \exp\left\{ \frac{T}\eps |V_s^\eps| \right\} \right] \D s.
\end{align*}
Next, by the properties of the logarithm and supremum
\begin{align*}
& \limsup_{\eps\downarrow 0}\eps\log\EE\left[ \exp\left\{ \frac 1\eps \int_0^T f(V_s^\eps)\D s \right\} \right] \\ & \qquad \qquad \leq \limsup_{\eps\downarrow 0}\left\{\eps\log\frac1T\right\} + T + \limsup_{\eps\downarrow 0}\eps\log\left\{\int_0^T  \EE\left[ \exp\left\{ \frac{T}\eps |V_s^\eps| \right\} \right] \D s\right\}.
\end{align*}
We can now apply Gronwall's Lemma to the last term, which yields for some $C>0$,
\[
\limsup_{\eps\downarrow 0}\eps\log\EE\left[ \exp\left\{ \frac 1\eps \int_0^T f(V_s^\eps)\D s \right\} \right] \leq T + \limsup_{\eps\downarrow 0}\eps\log\left\{ \exp\left\{ CT + \frac{T}{\eps} v_0 \right\} \right\} = T(1+v_0),
\]
which is finite.
\end{proof}
Following this lemma, the results from the previous section could simply be adapted so that $\{V_t^\eps,W_t^{\eps}\}$ satisfy the same LDP by simply setting $f=0$ (or equivalently $\kappa = 0$ in the case of Heston). However, this violates the condition $f(0)>0$ in Baldi and Caramellino~\cite{Baldi2011}. Fortunately, Conforti, Deuschel and De Marco~\cite{conforti2015small} removed the need for strict positivity on the drift at the initial time by imposing more stringent conditions on the diffusion.
\begin{assumption}\label{ass:small_time_LDP}\
\begin{enumerate}[i)]
	\item There exists $\xi>0$ such that  $g(y)=\xi |y|^\gamma$ for $\gamma \in[1 / 2,1)$ for all $y\geq0$;
	\item The equality $b(y)=\tau(y)+ K y$ holds for all $y\geq0$, where $\tau$ is a Lipschitz continuous and bounded function, and $\tau(y) \geq 0$ in a neighbourhood of the origin.
\end{enumerate}
\end{assumption}

\begin{theorem}[Theorem 1.1 in~\cite{conforti2015small}]
Under Assumption~\ref{ass:small_time_LDP}, 
the solution~$V^\eps$ to~\eqref{eq:log_price_dynNoise}
satisfies $\left\{V^{\eps}\right\}\sim\LDP(\IIV, \Cc([0, T]\to\RR^{+}))$ with
$$
\IIV(\varphi)=\frac{1}{2 \xi^{2}} \int_{0}^{T}\left(\frac{\dot{\varphi}_{t} - K\varphi_{t}}{\varphi_{t}^{\gamma}}\right)^{2} \ind_{\left\{\varphi_{t} \neq 0\right\}} \D t.
$$
\end{theorem}
Therefore by setting $K=0$ we can use the methodology form the previous section since the LDPs are the same. Similarly as before, we only consider the deterministic change of drift, since the stochastic case is very similar.
We therefore search for~$\hh$ such that
$$
\hh^* = \argmax\limits_{\xx\in \HH_T^0} 
\left\{F(\varphi_T(\xx)) - \frac{\|\dot{\xx}\|_{T}^2}{2}\right\},
$$
where $\varphi_t(\xx)$ is the unique solution on $[0,T]$ to:
\begin{equation*}
\left\{
\begin{array}{rll}
\varphi_t(\xx) & = \displaystyle \sqrt{\psi_t}\vrho \dot{\xx}(t)^\top, & \varphi_0(\xx) = 0,\\
\dot{\psi}_t & =  \displaystyle g(\psi_t) \dot{x}_1(t), & \psi_0 = v_0.
\end{array}
\right.
\end{equation*}

\subsubsection{Example: option with path-dependent payoff}
We consider a payoff $G(\alpha\cT X)$ 
as in Section~\ref{sec:PayoffPathDepG}.
In the deterministic case, we have
\begin{equation*}
\left\{
\begin{array}{rl}
\varphi^{*},\psi^{*}
& = \displaystyle \argmax\limits_{\{\varphi,\psi\} \in \HH_T^0\times \HH_T^{v_0,+}} \left\{F\left(\int_0^T\alpha_t \dot{\varphi}_t \D t\right) - \frac{1}{2}\int_0^T \left\{U(\psi_t)^2 + \left(\frac{Z(\varphi_t,\psi_t) - \rho U(\psi_t)}{\brho}\right)^2\right\}\D t\right\},\\
\dot{h}_1^{*} & = \displaystyle U(\psi^{*}), \\
\dot{h}_2^{*} & = \displaystyle  \frac{Z(\varphi^{*},\psi^{*}) - \rho U(\psi^{*})}{\brho},
\end{array}
\right.
\end{equation*}
where
$U(\psi) = \frac{\dot{\psi}}{\xi \sqrt{\psi}}$
and
$Z(\varphi,\psi) = \frac{\dot{\varphi}}{\sqrt{\psi}}$.
The same way as before we only look at the Heston model and we have for all $t\in[0,T]$,
with $A_t = \frac{U_t}{\sqrt{\psi_t}}$,
$$
\frac{1}{\brho^2}(Z_t - \rho U_t) = \frac{1}{2} \beta \alpha_t \psi_t
\qquad\text{and}\qquad
\dot{A}_t= -\frac{1}{2}\xi A_t^2 + \frac{1}{2}\xi \beta \alpha_t \left(\frac{1}{2} - \frac{1}{4}\brho^2\beta\alpha_t \right) + \frac{1}{2}\rho\beta\dot{\alpha}_t.
$$

\begin{remark}
Variance reduction for affine stochastic volatility processes via importance sampling through the large-time approximation is extensively covered in~\cite{Grbac2021}, so we do not repeat the study and refer the reader to the aforementioned work.
\end{remark}

\section{Importance sampling via moderate deviations}\label{sec:MDP_IS}
In the previous sections, 
large deviations provided us with a way of computing the asymptotic change of measure for importance sampling, 
via an $\eps$-approximation of the log-price~$X^{\eps}$.
While the large deviations rate function
was a convenient quadratic in the deterministic volatility setting, 
it is in general rather cumbersome to compute numerically, 
unfortunately offsetting any importance sampling gain. 
Moderate deviations act on a cruder scale, but provide quadratic rate functions, 
easier to compute.
Suppose that the sequence $\{X^{\eps}\}_{\eps>0}$ converges in probability to~$\overline{X}$. 
Moderate deviations for $\{X^{\eps}\}_{\eps>0}$ are defined as large deviations for 
the rescaled sequence 
$$
\left\{\frac{X^{\eps} - \overline{X}}{\sqrt{\eps}h(\eps)}
\right\}_{\eps>0},
$$
where $h(\eps)$ tends to infinity 
and $\sqrt{\eps}h(\eps)$ to zero
as~$\eps$ tends to zero. 
A typical choice is $h(\eps) = \eps^{-\alpha}$ for $\alpha \in (0,\frac{1}{2})$ or equivalently $\frac{1}{\sqrt{\eps}h(\eps)} = \eps^{-\alpha}$ for $\alpha\in(0,\frac{1}{2})$. 
We shall stick to this choice of~$h$ 
in our analysis in order to highlight clear rates of convergence.
We now introduce the approximation
\begin{equation}\label{eq:MDPSmallNoiseGeneral}
\widetilde{X}^{\eps} := \overline{X} + 
\eps^{-\alpha}\left(X^{\eps} - \overline{X}\right).
\end{equation}
This process is centered around~$\overline{X}$ and is a simple candidate. 
Furthermore, in stochastic volatility models, and particularly in large-time setting, the moderate deviations rate function is simply the second-order Taylor expansion of the large deviations rate function around its minimum~$\overline{X}$ ~\cite[Remark 3.5]{Jacquier2019}. 
We again consider the dynamics~\eqref{eq:log_price_dyn}
with Assumption~\ref{ass:SDE_coeffs} for the coefficients. 
We further assume the following conditions:

\begin{assumption}\label{ass:SDE_V}\ 
\begin{enumerate}[i)]
    \item For $f\in \Cc^2(\RR^{+}\to\RR)$, 
    the equation $\psi_t = v_0 + \int_0^t f(\psi_s)\D s$ admits a unique strictly positive solution $\psi\in\Cc^2([0,T]\to\R)$;
    \item The small-noise approximation~\eqref{eq:log_price_dynNoise} of~$V$ satisfies an LDP with the good rate function~$\IIV$ and speed~$\eps$ such that~$\IIV$ admits a unique minimum and is null there.
\end{enumerate}
\end{assumption}
As it will be shown in Lemma~\ref{lem:degenerate-limit}, 
the sequence $\{V^\eps\}_{\eps>0}$ converges in probability to the function~$\psi$ as a consequence of Assumption \ref{ass:SDE_V}. 
This provides a natural choice for the centered process 
${\overline{X}_t = -\frac{1}{2}\int_0^t\psi_s\D s}$,
so that the approximation~\eqref{eq:MDPSmallNoiseGeneral}
reads, for any $t\in [0,T]$,
\begin{equation}\label{eq:MDPSmallNoiseX}
\widetilde{X}^\eps_t = -\frac{1}{2} \int_0^t\psi_s \D s + \eps^{-\alpha}\left(X^{\eps}_t + \frac{1}{2} \int_0^t\psi_s \D s\right).
\end{equation}

\subsection{Small-noise moderate deviations}
Plugging in the small-noise approximation of~$X^\eps$ introduced in~\eqref{eq:log_price_dynNoise}, the process~$\widetilde{X}^\eps$ in~\eqref{eq:MDPSmallNoiseX} satisfies the SDE
\begin{align*}
\D \widetilde{X}^\eps_t
 & = \displaystyle\left(-\frac{1}{2} + \frac{1}{2\eps^{\alpha}}\right)\psi_t\D t
+ \eps^{-\alpha}\D X^{\eps}_t\\
 & = \displaystyle\left(-\frac{1}{2} + \frac{1}{2\eps^{\alpha}}\right)\psi_t\D t
+ \frac{1}{\eps^{\alpha}}\left(-\frac{1}{2}V_t^\eps \D t + \seps\sqrt{V_t^\eps} \D B_t\right)\\
 & = \displaystyle
 -\frac{1}{2}\psi_t\D t - \frac{1}{2\eps^{\alpha}}
 \left(V_t^\eps - \psi_t\right) \D t + 
\eps^{\half-\alpha}\sqrt{V_t^\eps} \D B_t
\end{align*}
starting from $\widetilde{X}^{\eps}_0 = 0$,
together with the small-noise approximation~\eqref{eq:log_price_dynNoise} for the variance:
$$
\D V_t^{\eps} = \displaystyle  f(V_t^{\eps})\D t + \seps g(V_t^{\eps})\D W_t,
$$
starting from $V_0^{\eps} = v_0$.
This transformation creates a lag between the decreasing speeds of~$\widetilde{X}^{\eps}$ and that of~$V^{\widetilde{\eps}}$ (speed of convergence to zero of the diffusion part of the volatility process i.e. $\eps^{\frac{1}{2} - \alpha}$ versus $\eps^{\frac{1}{2}}$). Since $\widetilde{X}^{\eps}$ is our reference, 
we adjust the speed of the LDP via 
$\eps^{1 - 2\alpha}\mapsto \seps$.
With $\beta := \alpha / (1 - 2\alpha)$
and $\eta^\eps_t := (V_t^\eps - \psi_t)\eps^{-\beta}$, we obtain the system
\begin{equation}\label{eq:eps_MDP__eta_V}
\def\arraystretch{1.25}
\begin{array}{rll}
\D \widetilde{X}^\eps_t & = \displaystyle -\frac{1}{2}\psi_t \D t -\frac{1}{2}\eta^\eps_t \D t + \sqrt{\eps}\sqrt{V_t^\eps} \D B_t, 
& \widetilde{X}_0^\eps = 0,\\
\D \eta_{t}^\eps & = \displaystyle \eps^{-\beta}(f(V_t^\eps)-f(\psi_t))\D t + \sqrt{\eps}g(V_t^\eps)\D W_t, & \eta_{0}^\eps = 0,\\
\D V_t^\eps & = \displaystyle  f(V_t^\eps)\D t + \eps^{\frac{1}{2} + \beta}g(V_t^\eps)\D W_t, & V_0^\eps = v_0.
\end{array}
\end{equation}
Similarly, in the price small-noise setting we have for $\gamma =\frac{1}{1-2\alpha} > 0$,
\begin{equation}\label{eq:eps_price_dyn}
\def\arraystretch{1.25}
\begin{array}{rll}
\D \widetilde{S}^\eps_t & = \displaystyle -\frac{1}{2}\eps^\gamma\eta^\eps_t \D t + \sqrt{\eps}\sqrt{V_t^\eps} \D B_t,
& \widetilde{S}_0^\eps = 1,\\
\D \eta_{t}^\eps & = \displaystyle \eps^{-\beta}(f(V_t^\eps)-f(\psi_t))\D t + \sqrt{\eps}g(V_t^\eps)\D W_t, & \eta_{0}^\eps = 0,\\
\D V_t^\eps & = \displaystyle  f(V_t^\eps)\D t + \eps^{\frac{1}{2} + \beta}g(V_t^\eps)\D W_t, & V_0^\eps = v_0.
\end{array}
\end{equation}
In the following, we provide an LDP for~$\{\eta^\eps\}$ 
(equivalently an MDP for~$\{V^\eps\}$). 
We relegate more technical proofs to Appendix~\ref{apx:proofs}.

\subsubsection{Theoretical results}
The main moderate deviations result of this section is Theorem~\ref{thm:mdp-small-noise}, 
but we first start with the following three technical lemmata,
useful for the theorem but also of independent interest, proved in Appendices~\ref{sec:lem:degenerate_Proof}-\ref{sec:lem:expct_convg_Proof}-\ref{sec:lem:degenerate-limit_Proof}:
\begin{lemma}\label{lem:degenerate}
Let $\{Z^{\eps}\}_{\eps>0}$ be a family of random variables mapping to any metrisable space~$\Xx$ and satisfying an LDP with good rate function~$\II$.
If there exists a unique $x_0$ such that $I(x_0) = 0$, 
then for all $\beta > 1$, $Z^{\eps^\beta}$ satisfies an LDP with the good rate function
\[ 
\II(x)=
\begin{cases} 
      0, & \text{ for } x=x_0, \\
      +\infty, & \text{ elsewhere.}
\end{cases}
\]
Equivalently, if for $\beta > 1$, $Z^{\eps}$ satisfies an LDP with speed $\eps^{\beta}$ and the good rate function with a unique minimum at zero, then $Z^{\eps}$ is exponentially equivalent to $x_0$ with speed $\eps$.
\end{lemma}
As a consequence of this lemma, the sequence~$\{Z^{\eps}\}$ converges in probability to~$x_0$. 

\begin{lemma}\label{lem:expct_convg}
Let $\{Z^{\eps}\}_{\eps>0}$ be a sequence of random variables mapping to a metrisable space~$\Xx$ and satisfying an LDP with good rate function~$\II$ such that $\II(x)=0$ if and only if $x = x_0$ for some $x_0 \in\Xx$. 
If $\EE[Z^{\eps}]$ is uniformly integrable, then $\lim_{\eps\downarrow 0}\EE[Z^{\eps}] = x_0$.
\end{lemma}

\begin{lemma}\label{lem:degenerate-limit}
Let $V^{\eps}$ be given by~\eqref{eq:eps_MDP__eta_V}.
such that $f\in\Cc(\RR^+\to \RR)$ and $g\in\Cc(\RR^+\to \RR^+)$ satisfy Assumption~\ref{ass:SDE_coeffs}. Then $\{V^{\eps}\}$ converges almost surely to the unique solution of $\dot{\psi}_t = f(\psi_t)$ on $[0,T]$ with boundary condition $\psi(0) = v_0$.
\end{lemma}


\begin{theorem}\label{thm:mdp-small-noise}
Let $\beta, v_0>0$ and let $f, g \in \Cc_T$ be such that Assumption~\ref{ass:SDE_coeffs} is satisfied.
Under Assumption~\ref{ass:SDE_V}, 
let $W^{\eps} := \sqrt{\eps}W$ and $V^\eps, \eta^{\eps}$ defined in~\eqref{eq:eps_MDP__eta_V}.
The triple
$\{V^\eps, \eta^\eps,W^\eps\}$ satisfies an LDP with speed~$\eps$ and the good rate function
\begin{equation*}
\II^{V,\eta,W}(v,\eta,w) = 
\left\{
\begin{array}{ll}
\displaystyle \frac{1}{2}\|\dot{w}\|_{T}^2, & \text{if }
\displaystyle 
w \in \HH_T^0,\;
v = \psi,\;
\dot{\eta} = \dot{f}(\psi)\eta + g(\psi)\dot{w},\\
+\infty, & \text{otherwise}.
\end{array}
\right.
\end{equation*}
\end{theorem}

\subsubsection{Importance sampling using MDP}
Consider the system~\eqref{eq:eps_MDP__eta_V}, 
and let $\Wf^\eps \coloneqq \sqrt{\eps}\Wf$ and $\Yf^\eps \coloneqq \sqrt{V^\eps}\circ \Wf^\eps$. 
Following Theorem~\ref{thm:mdp-small-noise} and Lemma~\ref{thm:robertsonLDP}, and the Contraction principle and \cite[Exercise 4.2.7]{Dembo2010} imply that the triple $\{\eta^\eps,\Wf^\eps, \Yf^\eps\}$ satisfies an LDP with speed $\eps$ and good rate function 
\begin{equation*}
\II^{\eta,\Wf,\Yf}(\eta,\ww, \yy) =
\left\{
\begin{array}{ll}
\displaystyle 
\frac{1}{2}\|\dot{\ww}\|_{T}^2, & \text{if }
\displaystyle 
\ww \in \HH_T^0,\;
\eta = \int_{0}^{\cdot}\left\{\dot{f}(\psi_s)\eta_s + g(\psi_s)\dot{w}_1(s)\right\}\D s,\;
\yy = \sqrt{\psi}\circ \ww,\\
+\infty, & \text{otherwise}.
\end{array}
\right.
\end{equation*}


\subsubsection{Log-price small-noise MDP}
Consider a continuous payoff function $G\in\Cc([0,T]\to\R^{+}$ and let $F:=\log|G|$. 
As a reminder, we are interested in finding a measure change minimising $\EE[\E^{2F}\frac{\D\PP}{\D\QQ}]$.
We first consider a deterministic change of drift, via the change of measure 
\[
\left.\frac{\D\QQ}{\D\PP}\right\vert_{\Ff_T} = \exp\left\{-\frac{1}{2}\|\dot{\hh}\|_{T}^2 + \dot{\hh}\cT \Wf^\top\right\},
\]
for $\hh \in \HH_T^0$ with $\dot{\hh}$ of finite variation. 
In the spirit of moderate deviations, we use the approximation 
\[
\left.\frac{\D\QQ}{\D\PP}\right\vert_{\Ff_T} \approx \exp\left\{-\frac{1}{2}\|\dot{\hh}\|_{T}^2 + \dot{\hh}\cT \Wf^{\eps\top}\right\},
\]
and thus aim at minimising
\[
L(\hh) = \limsup\limits_{\eps \downarrow 0} \eps\log\EE\left[\exp\left\{\frac{1}{\eps}\left(
2F(\widetilde{X}^\eps) + \frac{1}{2}\|\dot{\hh}\|_{T}^2 - \dot{\hh}\cT \Wf^{\eps\top}\right)\right\}\right].
\]
Under the conditions of Varadhan's lemma~\ref{thm:varadhan_modified} (e.g. if~$F$ is bounded), then
$L(\hh) = 
\sup_{\xx \in \HH_T^0} 
\Lf(\xx,\hh)$
where
$$
\Lf(\xx,\hh)
 = 2 F\left(-\frac{1}{2}\int_0^.(\psi_t + \eta_t)\D t + \vrho\Yf^\top \right) - \dot{\hh}\cT\xx^\top + \frac{\|\dot{\hh}\|_{T}^2}{2} - \frac{\|\dot{\xx}\|_{T}^2}{2},
 $$
with $\vrho \coloneqq \begin{pmatrix}\rho &\brho\end{pmatrix}^\top$ and 
\begin{equation}\label{eq:etaY1Y2}
\eta(\xx) = \int_{0}^{\cdot}\left\{\dot{f}(\psi_s)\eta_s + g(\psi_s)\dot{x}_1(s)\right\}\D s
\qquad\text{and}\qquad
\Yf(\xx) = \sqrt{\psi}\circ \xx.
\end{equation}
Minimising $L(\hh)$ is far from trivial, 
and hence, as before, we define the optimal change of drift $\hh^{*}$ as a solution to the dual problem
$\inf\limits_{\hh \in \HH_T^0} \Lf(\xx,\hh)$,
so that, with~$\eta$ as in~\eqref{eq:etaY1Y2},
\begin{equation}\label{eq:optidetdrift}
\begin{cases}
&\displaystyle \Yf = \sqrt{\psi}\circ \xx \\
&\displaystyle \psi_t = v_0 + \int_0^t f(\psi_s)\D s, \\
&\displaystyle \eta_t = \int_0^t\left\{\dot{f}(\psi_s)\eta_s + g(\psi_s)\dot{x}_1(s)\right\}\D s, \\
&\displaystyle \hh^{*} = \argmax\limits_{\xx \in \HH_T^0} \left\{F\left(-\frac{1}{2}\int_{0}^{\cdot}(\psi_t + \eta_t)\D t + \vrho\Yf^\top\right)- \frac{\|\dot{\xx}\|_{T}^2}{2}\right\}.
\end{cases}
\end{equation}

\begin{remark}
This moderate deviations approach is equivalent to approximating $\sqrt{V_t}$ with $\sqrt{\psi_t}$ and $V_t$ with some Gaussian process centered at $\psi$.
\end{remark}

We now consider a stochastic change of drift, 
through the Radon-Nikodym derivative
\[
\left.\frac{\D\QQ}{\D\PP}\right\vert_{\Ff_T} = \exp\left\{-\frac{1}{2}\|\dot{\hh}\sqrt{V}\|_{T}^2 + (\dot{\hh}\sqrt{V})\cT\Wf^\top\right\},
\]
for $\hh \in \HH_T^0$, $\dot{\hh}$ of finite variation and such that $\EE\left[\frac{\D\QQ}{\D\PP}\right]=1$. 
Again, we use the approximation
\[
\left.\frac{\D\QQ}{\D\PP}\right\vert_{\Ff_T} \approx \exp\left\{-\frac{1}{2}\|\dot{\hh}\sqrt{V^\eps}\|_{T}^2 + (\dot{\hh}\sqrt{V^\eps})\cT({\Wf^\eps})^\top\right\}
\]
and aim at minimising
\[
L(\hh) = \limsup\limits_{\eps \downarrow 0} \eps\log\EE\left[\exp\frac{1}{\eps}\left\{
2F(\widetilde{X}^\eps) +\frac{1}{2}\|\dot{\hh}\sqrt{V^\eps}\|^2_T - \left (\dot{\hh}\sqrt{V^\eps}\right )\cT \Wf^{\eps\top}\right\}\right].
\]
If Varadhan's lemma conditions hold, then again
$L(\hh) = \sup_{\xx \in \HH_T^0}\Lf(\xx,\hh)$
where
$$
\Lf(\xx,\hh) = 2F\left(-\frac{1}{2}\int_{0}^{\cdot}(\psi_t + \eta_t)\D t + \vrho \Yf^\top\right) - \left (\dot{\hh}\sqrt{\psi}\right )\cT \xx^\top + \frac{\|\dot{\hh}\sqrt{\psi}\|_{T}^2}{2}
- \frac{\|\dot{\xx}\|_{T}^2}{2},
$$
with~$\eta$ defined as in~\eqref{eq:etaY1Y2}. As minimising $L(\hh)$ is \textit{a priori} complicated, we define our optimal change of drift $\hh^{*}$ as a solution to the dual problem
\[
\begin{cases}
&\displaystyle \xx^{*} = \argmax\limits_{\xx \in \HH_T^0} \left\{
F\left(-\frac{1}{2}\int_{0}^{\cdot} \left(\psi_t + \eta_t\right)\D t - \vrho \Yf^\top \right)
- \frac{\left\|\dot{\xx}\right\|_{T}^2}{2}
\right\},\\
&\displaystyle \psi_t = v_0 + \int_0^t f(\psi_s)\D s, \\
&\displaystyle \eta_t = \int_0^t\left\{\dot{f}(\psi_s)\eta_s + g(\psi_s)\dot{x}_1(s)\right\}\D s, \\
&\displaystyle  \hh^{*} = \int_{0}^{\cdot}\frac{\dot{\xx}^{*}(t)}{\sqrt{\psi_t}}\D t.
\end{cases}
\]

\subsubsection{Price small-noise MDP}
We consider now the stock price dynamics given in~\eqref{eq:eps_price_dyn}:
$$
\D\widetilde{S}^\eps_t = -\frac{1}{2}\eps^\gamma\eta_{t}^\eps\D t + \sqrt{\eps}\sqrt{V_t}\D B_t,
$$
starting from $\widetilde{S}_0^\eps = 1$, with $\gamma >0$.
With the  deterministic change of drift
$$
\left.\frac{\D\QQ}{\D\PP}\right\rvert_{\Ff_T}
= \exp\Big\{-\frac{1}{2}\|\dot{\hh}\|_{T}^2 + \dot{\hh}\cT \Wf^\top\Big\},
$$
for $\hh \in \HH_T^0$ with $\dot{\hh}$ of finite variation,
the optimal change of drift $\hh^{*}$ is the solution to
\begin{equation}\label{eq:optidetdrift-price}
\hh^{*} = \argmax\limits_{\xx \in \HH_T^0} 
\left\{F(\vrho \Yf^\top)- \frac{\|\dot{\xx}\|_{T}^2}{2}\right\}.
\end{equation}

Regarding the stochastic change of drift,
the Radon-Nikodym derivative takes the form
\[
\left.\frac{\D\QQ}{\D\PP}\right\rvert_{\Ff_T} = \exp\Big\{-\frac{1}{2}\|\dot{\hh}\sqrt{V}\|_{T}^2
 + (\dot{\hh}\sqrt{V})\cT \Wf^\top\Big\}
\]
with $\dot{\hh}$ of finite variation and such that $\EE\left[\frac{\D\QQ}{\D\PP}\right]=1$.
We define our optimal change of drift $\hh^{*}$ as a solution to the dual problem:
$$
\hh^{*} = \int_{0}^{\cdot}\frac{\dot{\xx}^{*}(t)}{\sqrt{\psi_t}} \D t,
\qquad\text{with}\qquad 
\xx^{*} = \argmax\limits_{\xx \in \HH_T^0} \left\{F(\vrho \Yf^\top)
- \frac{\|\dot{\xx}\|^2_T}{2}\right\}.
$$
with~$\eta$, $\Yf$ as in~\eqref{eq:etaY1Y2}. Again the objective simplifies to the case of the deterministic drift change, the only difference being the way $\hh^{*}$ is calculated.


\subsubsection{Example: options with path-dependent payoffs}
We apply our methodology to options with payoffs of the form $G(\alpha \cT X)$, where $G:\R \rightarrow \R^{+}$ is a differentiable function and $\alpha$ a positive (almost everywhere) function of class $\Cc^1([0,T]\to\R^{+})$. The payoff is then a continuous function of the path. 
Let $F = \log|G|$ and $\overline{F}(x) = F(x-\frac{1}{2}\int_0^T\psi_s\alpha_s \D s)$ and suppose that Assumptions~\ref{ass:SDE_coeffs} and~\ref{ass:SDE_V} hold.

\subsubsection{Log-price small-noise MDP with path-dependent payoff}

We consider a deterministic change of drift and proceed similarly to Section~\ref{sec:LDP_IS} 
with the transformations
$$
\dot{\phi} = \sqrt{\psi}\left( \rho\dot{x}_1 + \brho\dot{x}_2 \right) - \frac{\eta}{2} \qquad \text{and} \qquad \dot{\eta} = \dot{f}(\psi)\eta + g(\psi)\dot{x}_1,
$$
so that the optimisation problem~\eqref{eq:optidetdrift} for a path-dependent payoff can then be written as
\begin{equation}\label{eq:optidetdriftapplication}
\left\{
\begin{array}{rl}
\phi^{*}, \eta^{*} & = \displaystyle  \argmax\limits_{\phi,\eta \in \HH_T^0} \left\{ \overline{F}\left(\int_0^T \alpha_t \dot{\phi}_t \D t \right) 
- \half\int_0^T \left\{U(\eta)^2 + \frac{[Z(\phi_t,\eta_t) - \rho U(\eta_t)]^2}{1-\rho^2}\right \}\D t \right\},\\
\dot{\hh}^{*}  & = \displaystyle 
\left(U(\eta^{*}),  \frac{1}{\brho}\left(Z(\phi^{*}, \eta^{*}) - \rho U(\eta^{*})\right)\right),
\end{array}
\right.
\end{equation}
with
$$
U(\eta) = \frac{\dot{\eta} - \dot{f}(\psi)\eta}{g(\psi)}
\qquad\text{and}\qquad
Z(\phi, \eta) = \frac{\dot{\phi} + \frac{1}{2}\eta}{\sqrt{\psi}}. 
$$
Then, by applying Euler-Lagrange to the problem seen as an optimisation over $\left \{\int_{0}^{\cdot} \alpha_t\dot{\phi}_t\D t,\eta\right \}$, we obtain the system of ODEs
\begin{equation*}
\left\{
\begin{array}{rcl}
Z - \rho U &=& \displaystyle \frac{1}{2}\beta\brho^2\alpha \sqrt{\psi},\\
\displaystyle \frac{\D}{\D t}\left\{\frac{2U}{g(\psi)} - \beta \rho \alpha \frac{\sqrt{\psi}}{g(\psi)}\right\} &=& \displaystyle \beta \alpha \sqrt{\psi} \left(\frac{1}{2\sqrt{\psi}} + \rho \frac{\dot{f}(\psi)}{g(\psi)} \right) -2U\frac{\dot{f}(\psi)}{g(\psi)},
\end{array}
\right.
\end{equation*}
with boundary conditions $\beta = - 2\overline{F}_T'$ and $U_T = -\rho\alpha_T\sqrt{\psi_T}\overline{F}_T'$, where $\overline{F}_T':= \overline{F}'(\int_0^T\alpha_t\dot{\phi}_t\D t)$. 
Introducing
$A := \frac{2U}{g(\psi)} - \beta \rho \alpha \frac{\sqrt{\psi}}{g(\psi)}$
simplifies the problem to the linear ODE:
$\dot{A}- \dot{f}A = \frac{1}{2}\beta\alpha$ with $A_T = 0$,
with solution
\[
A_t = \frac{1}{2}\beta \E^{-B_t}\left( \int_0^t \E^{B_s} \alpha_s \D s - \gamma_T \right),
\quad \text{for } t\in[0,T]\,,
\]
where $B:=\int_{0}^{\cdot} f(\psi_t)\D t$ and $\gamma := \int_{0}^{\cdot} \E^{B_t}\alpha_t \D t$. 
We can now solve for~$U$ and~$Z$:
\begin{align*}
U = \beta\left( u + \frac{1}{2}\rho\alpha\sqrt{\psi} \right) \qquad \text{and} \qquad
Z = \beta\left ( \rho u + \frac{1}{2}\alpha\sqrt{\psi} \right )
\end{align*}
with $u = \frac{1}{4}g(\psi)\E^{-B}(\gamma-\gamma_T)$. Our optimisation problem was posed over $\left \{\int_{0}^{\cdot} \alpha_t \dot{\phi}_t \D t, \eta\right \}$ so we require the solution in terms of this couple, therefore:
\begin{align*}
\text{i)}& \;\; \int_0^t \alpha_s \dot{\phi}_s \D s = \beta\int_0^t \alpha_s \left \{ \dot{\phi}^1_s - \frac{1}{2}\dot{\phi}^2_s \right \}\D s + \eta_0 T \E^{B_T} \\
& \qquad \qquad  \dot{\phi}^1_t = (u+ \frac{1}{2}\alpha \sqrt{\psi})\sqrt{\psi} \qquad \qquad \dot{\phi}^2 = \E^B \int_{0}^{\cdot} \E^{-B_s}\left( u_s + \frac{1}{2}\alpha_s\sqrt{\psi_s} \right )g(\psi_s)\D s \\
\text{ii)}& \;\; \eta_t = \beta \E^{B_t} \int_0^t \E^{-B_s}\left( u_s + \frac{1}{2}\rho\alpha_s\sqrt{\psi_s} \right )g(\psi_s)\D s + \eta_0 \E^{B_t},
\end{align*}
where $\beta,\eta_0\in\RR$ are parameters over which we perform our optimisation. 
Thus the original optimisation objective~\eqref{eq:optidetdriftapplication} becomes 
\begin{align*}
\beta^{*}, \eta_0^{*} = \argmax_{\beta,\eta_0\in\RR}
\left\{\overline{F}\left( \beta\int_0^t \alpha_s \left \{ \dot{\phi}^1_s - \frac{1}{2}\dot{\phi}^2_s \right \}\D s + \eta_0 T \E^{B_T} \right ) - \frac{1}{2}\beta^2\int_0^T\left\{ \big| \dot{\phi}^2_t \big|^2 + \frac{1}{4}\brho^2\alpha_t^2\psi_t \right\}\D t\right\}.
\end{align*}

\begin{remark}
Stochastic change of drift objective is equivalent to the one with deterministic drift, 
the difference being how~$\hh$ is calculated.
\end{remark}

\subsubsection{Price small-noise MDP with path-dependent payoff}

We consider a deterministic change of drift.
When considering price dynamics, the $\eta$ does not play a role any more, but we can nevertheless perform a similar transformation as before. Let ${\dot{\phi} = \sqrt{\psi}\left( \rho \dot{x}_1 + \brho\dot{x}_2 \right)}$ and
$\dot{\varphi} = \brho\sqrt{\psi}\dot{x}_2$,
so that~\eqref{eq:optidetdrift-price} becomes
\begin{equation}\label{eq:optidetdriftapplication-price}
\left\{
\begin{array}{rcl}
\phi^{*}, \varphi^{*} &=& \displaystyle 
\argmax\limits_{\phi,\varphi \in \HH_T^0} \left\{ \overline{F}\left(\int_0^T \alpha_t \dot{\phi}_t \D t \right)- \frac{1}{2} \int_0^T \left\{U(\varphi)^2 + \frac{1}{\brho^2}\left(Z(\phi,\varphi) - \rho U(\varphi)\right)^2\right \}\D t \right\} \\
\left(\dot{h}^{*}_1,  \dot{h}^{*}_2\right) &=& 
\displaystyle \left(U(\varphi^{*}), 
\frac{1}{\brho}\left(Z(\phi^{*}, \varphi^{*}) - \rho U(\varphi^{*})\right)\right),
\end{array}
\right.
\end{equation}
with
$U(\varphi) = \frac{\dot{\varphi}}{\brho\sqrt{\psi}}$ and
$Z(\phi, \varphi) = \frac{\dot{\phi}-\dot{\varphi}}{\sqrt{\psi}}$.
Now seen as optimisation over $\left \{\int_{0}^{\cdot} \alpha_t \dot{\phi}_t \D t, \eta\right \}$ we obtain by Euler-Lagrange the system of ODEs 
\begin{equation*}
\left\{
\begin{array}{rcl}
Z - \rho U &=& \displaystyle \frac{1}{2}\beta\brho^2\alpha \sqrt{\psi}, \\
\displaystyle\frac{\D}{\D t}\left\{\frac{U}{\sqrt{\psi}} - \frac{\brho+1}{\brho^2} (Z - \rho U) \right \} &=& 0,
\end{array}
\right.
\end{equation*}
with boundary condition $\beta = - 2\overline{F}_T'$ with $\overline{F}_T'$ as above,
which can be solved as
\begin{align*}
U = \frac{1+\brho}{2}\beta\alpha\sqrt{\psi}
\qquad \text{and} \qquad 
Z=\frac{\brho + \rho(\brho+1)}{2}\beta\alpha\sqrt{\psi},
\end{align*}
or
\begin{align*}
\dot{\varphi} = \frac{\brho(1+\brho)}{2}\beta\alpha\psi
\qquad \text{and} \qquad
\dot{\phi} = \frac{\varrho}{2}\beta\alpha\psi,
\end{align*}
where $\varrho := \brho + (\brho+1)(\rho+\brho)$. The optimisation problem~\eqref{eq:optidetdrift-price} thus simplifies to
\[
\beta^{*} = \argmax_{\beta\in\RR} 
\left\{\overline{F}\left(
\frac{\varrho\beta}{2}\int_0^T\alpha_t\psi_t\D t \right )
- \frac{\beta^2}{8}\left((\brho+1)^2+\brho^2\right)\int_0^T\alpha_t^2\psi_t^2 \D t\right\}.
\]

\begin{remark}
Stochastic change of drift objective is again equivalent to the one with deterministic drift the difference 
being how~$\hh$ is calculated.
\end{remark}

\subsection{Small-time moderate deviations}
We now mimic the results of the previous section, bu for small-time moderate deviations. 
Consider the log-price dynamics~\eqref{eq:log_price_dyn} under  Assumption~\ref{ass:SDE_coeffs} and Assumption~\ref{ass:SDE_V}.
Let $\alpha \in (0,\frac{1}{2})$ and 
${\widetilde{X}^\eps := \eps^{-\alpha} X^{\eps}}$, so that
\begin{equation*}
\begin{array}{rll}
\D\widetilde{X}^\eps_t
&= \displaystyle -\frac{1}{2}\eps^{1-\alpha} V_t^{\eps} \D t + \eps^{\frac{1}{2}-\alpha}\sqrt{V_t^\eps}\D B_t, 
& \widetilde{X}_0^\eps = 0, \\
\D V_t^\eps &= \displaystyle\eps f(V_t^\eps)\D t + \sqrt{\eps}g(V_t^\eps)\D W_t,  & V_0^\eps = v_0.
\end{array}
\end{equation*}
As we will see, as a consequence of Theorem~\ref{mdp-small-time} the results remain the same as in the case of price small-noise moderate deviations of the previous section with $f = 0$ and $\psi = v_0$. This being the case, we do not repeat them here. Let us nevertheless note that in the case of change of measure with deterministic drift, the problem is similar to the one, where $V$ is constant equal to $v_0$.
\subsubsection{Theoretical results}

Analogous to Lemma~\ref{lem:degenerate-limit} and Theorem~\ref{thm:mdp-small-noise}, we have 
\begin{lemma} \label{lem:degenerate-limit-small}
Let $V^{\eps}$ be given by~\eqref{eq:eps_MDP__eta_V}.
such that $f\in\Cc(\RR^+\to \RR)$ and $g\in\Cc(\RR^+\to \RR^+)$ satisfy Assumption~\ref{ass:SDE_coeffs}. 
Then $\{V^{\eps}\}$ converges in probability to~$v_0$.
\end{lemma}

\begin{theorem}\label{mdp-small-time}
Let $\beta, v_0>0$ and let $f, g \in \Cc_T$ be such that Assumption~\ref{ass:SDE_coeffs} is satisfied.
Under Assumption~\ref{ass:SDE_V}, 
let $W^{\eps} := \sqrt{\eps}W$ and $V^\eps, \eta^{\eps}$ defined in~\eqref{eq:eps_MDP__eta_V}.
The triple
$\{V^\eps, \eta^\eps,W^\eps\}$ satisfies an LDP with speed~$\eps$ and the good rate function
\begin{equation*}
\II^{V,\eta,W}(v,\eta,w) = 
\left\{
\begin{array}{ll}
\displaystyle \frac{1}{2}\|\dot{w}\|_{T}^2, & \text{if }
\displaystyle 
w \in \HH_T^0,\;
v = \psi,\;
\dot{\eta} =  g(\psi)\dot{w},\\
+\infty, & \text{otherwise}.
\end{array}
\right.
\end{equation*}
\end{theorem}

\subsection{Large-time moderate deviations} \label{sec:largetimeMDP}

We now consider a rescaling of ~\eqref{eq:log_price_dyn}
defined as $V_t^\eps = V_{\frac{t}{\eps}}$ and $X_t^\eps = \eps X_{\frac{t}{\eps}}$, so that under Assumption~\ref{ass:SDE_coeffs} and Assumption~\ref{ass:SDE_largetime_MDP},
\begin{equation}\label{eq:large_time_MD_dynamics}
\left\{
\begin{array}{rll}
\D X^\eps_t &= \displaystyle -\frac{1}{2}V_t^\eps \D t + \sqrt{\eps}\sqrt{V_t^\eps}\D B_t, 
& X_0^\eps = 0,\\
\D V_t^\eps &= \displaystyle \frac{1}{\eps}f(V_t^\eps)\D t + \frac{1}{\sqrt{\eps}}g(V_t^\eps)\D W_t, 
& V_0^\eps = v_0>0,
\end{array}
\right.\end{equation}
which leads to Regime 2 in the slow-fast setting of~\cite[Theorem 2.1]{Morse2017}, 
by choosing the time-scale separation parameter equal to~$\eps$. The following assumption is needed in order to conform to the conditions in~\cite{Morse2017}.
\begin{assumption}\label{ass:SDE_largetime_MDP}\
\begin{enumerate}[(i)]
    \item $f$ is locally bounded and of the form $f(y) = -\kappa y +\tau(y)$ with $\tau$ globally Lipschitz with Lipschitz constant~$L_{\tau}<\kappa$. In addition, the tail condition $\lim_{|y| \uparrow\infty}\frac{\tau(y)y}{|y|^2} = \kappa$ holds.
    \item The function $g$ is either uniformly continuous and bounded from above and away from zero or takes the form $g(y) = \xi |y|^{q_g}$ for $q_g \in [\frac{1}{2},1)$ with $\xi\ne 0$.
\end{enumerate}
\end{assumption}
\begin{remark}
Together with Assumption~\ref{ass:SDE_coeffs}, 
Condition~(ii), necessary to ensure ergodicity of the volatility process, collapses~$g$ to the form $g(y)=\xi |y|^\frac{1}{2}\,$ (for details we refer to~\cite{Jacquier2019,Morse2017}).
\end{remark}

\subsubsection{Theoretical results}
In order to apply the methodology from the previous sections to derive the desired changes of measure, we need a large-time MDP. More precisely, we need an MDP for $\{X^{\eps},\sqrt{\eps}W,\sqrt{\eps}\Wp\}$ in the case of deterministic drift. We do not consider the stochastic drift change here, since a rigorous treatment is out of scope of this paper.
Similar problem has been studied in~\cite[Theorem 2.1]{Morse2017} and~\cite[Theorem 3.3]{Jacquier2019}, where authors propose fewer conditions, although in a simpler setting (which happens to include the Heston model as well). 
We now introduce a theorem, which is a direct application of~\cite[Theorem~3.3]{Jacquier2019} and~\cite[Theorem~2.1]{Morse2017}, that provides the desired MDP.

\begin{theorem}\label{thm:limitPoisson}
Let $\Ll_{V}$ denote the infinitesimal generator of~$V$ before rescaling, i.e.
\[
\Ll_V h = f\dot{h} + \frac{1}{2}g^2\ddot{h}.
\]
Under Assumption~\ref{ass:SDE_coeffs} and Assumption~\ref{ass:SDE_largetime_MDP} the following hold:
\begin{enumerate}[i)]
    \item There exists a unique invariant measure~$\mu$ corresponding to $\Ll_V$;
    \item The process~$X^{\eps}$ converges in probability to $-\frac{1}{2}\overline{v}t$, where $\overline{v} = \int_0^\infty y \, \mu(dy)$;
    \item There exists a unique solution $\varphi$ with at most polynomial growth to the Poisson equation
    \[
    \Ll_V(\varphi)(y) = \frac{y - \bar{v}}{2}, 
    \qquad \text{with} \qquad \int_0^\infty \varphi(y)\mu(dy)=0.
    \]
\end{enumerate}
Furthermore, let us denote $ Q=\int_0^\infty\boldsymbol\alpha\boldsymbol\alpha^\top\mu(\D y)$, where
$$
\boldsymbol\alpha =
\begin{pmatrix}
\rho\sqrt{y} + \dot{\varphi}(y)g(y) & \brho\sqrt{y}\\
1 & 0 \\
0 & 1 \\
\end{pmatrix}
$$
Then the triple $\{X^{\eps},\sqrt{\eps}W,\sqrt{\eps}\Wp\}$ follows an MDP with good rate function~$I_{Q}$,
where
$$
\II_{Q}(\phi) := 
\inf \left\{\frac{1}{2}\int_0^T u_s^\top u_s\D s: u\in L^2\left([0,T]; \RR^3\right), \, \dot{\phi}^\top Q^{-1} \dot{\phi} = u^\top u \right\},
$$
if $\phi\in\mathcal{AC}$ and infinite otherwise.
\end{theorem}
\begin{lemma}\label{q-invertible}
The matrix~$Q$ is invertible. 
\end{lemma}
\begin{proof}
Let~$Y$ a random variable with distribution~$\mu$, the invariant measure. Then 
\begin{align*}
\mathrm{det}(Q) &= \EE_\mu\left[ \left( \rho\sqrt{Y} + \dot{\varphi}(Y)g(Y)\right)^2 +\brho^2 Y \right ]
- \EE_\mu\left[ \rho\sqrt{Y} + \dot{\varphi}(Y)g(Y) \right ]^2
- \EE_\mu\left[ \brho \sqrt{Y} \right ]^2 \\
&= \EE_\mu\left[ \left( \rho\sqrt{Y} + \dot{\varphi}(Y)g(Y)\right)^2 \right ] - \EE_\mu\left[ \rho\sqrt{Y} + \dot{\varphi}(Y)g(Y) \right ]^2 + \EE_\mu\left [ \brho^2 Y\right ] - \EE_\mu\left[ \brho \sqrt{Y} \right ]^2
\geq 0,
\end{align*}
by Cauchy-Schwarz, where an equality would imply that~$V$ is constant $\mu$-almost surely. 
This implies $f(\overline{v})=g(v)=0$, which is not possible due to Assumption~\ref{ass:SDE_coeffs} on $f$ and $g$. 
\end{proof}

\subsubsection{Example: Options with path-dependent payoff}
Consider again an option with a continuous payoff $G(X)$ and let $F=\log|G|$. Following the approach of the previous sections and using similar notations, we search solutions to the dual problems for deterministic change of drift.
Let 
$\xx \coloneqq \begin{pmatrix}x_1& x_2& x_3\end{pmatrix}^\top$ and 
$\underline{\xx}\coloneqq \begin{pmatrix}x_2 &  x_3\end{pmatrix}^\top$. 
By considering a deterministic change of drift similarly as before our problem becomes 
\[
L(\hh) = \limsup\limits_{\eps \downarrow 0} \eps\log\EE\Big[\exp\frac{1}{\eps}\Big\{
2F(X^\eps) +\frac{1}{2}\|\dot{\hh}\|^2_T - \dot{\hh}\cT({\Wf^\eps})^\top\Big\}\Big]
\]
Applying the modified Varadhan's lemma gives us the 
target functional
$$
\Lf(\xx,\hh) = 2F(x_1) + \frac{1}{2}\|\dot{\hh}\|^2_T - \dot{\hh}\circ\underline{\xx}^\top - \frac{1}{2}\int_0^T\dot{\xx}^\top Q^{-1} \dot{\xx} \D t.
$$
We now write the problem in terms of~$\widetilde{\xx}:=\dot{\xx}\,$:
$$
\begin{cases}
& \displaystyle \widetilde{\xx}^{*}  =  \argmax\limits_{\widetilde{\xx} \in L^2([0,T]; \RR^3)}
\left\{F\left(\int_{0}^{\cdot} \widetilde x_1(t) \D t\right) - \frac{1}{4}\int_0^T \widetilde{\xx}^\top(t) \left( Q^{-1} + \mathrm{diag}(0,1,1) \right )\widetilde{\xx}(t) \D t\right\}, \\
& \displaystyle \dot{\hh} = \widetilde{\underline{\xx}}^{*},
\end{cases}
$$
which is equivalent to
\begin{equation}\label{eq:SupInfInt}
\sup\limits_{x_1\in L^2([0,T]; \RR)}\left \{F\left(\int_{0}^{\cdot} \widetilde x_1(t) \D t\right) - \frac{1}{4}\int_0^T \inf_{\underline{\xx}\in L^2([0,T]; \RR^2)} \xx^\top(t) \underbrace{\left( Q^{-1} + \mathrm{diag}(0,1,1) \right )}_{\eqqcolon A}\xx(t) \D t \right \}.
\end{equation}
This substantially simplifies the main optimisation problem, which can now be solved explicitly, by first denoting $A \coloneqq Q^{-1} + \mathrm{diag}(0,1,1)$ and then writing 
\[
A =
\begin{pmatrix}
a_{11} & \mathbf a_{21}^\top\\
\mathbf a_{21} &A_{22}\\
\end{pmatrix}
\]
with $a_{11} \in\RR$, $\mathbf a_{21}\in\RR^{2}$, $A_{22}\in\RR^{2\times 2}$, so that the infimum in~\eqref{eq:SupInfInt} yields\footnote{Since $Q$ is positive-definite, both $A$ and $A_{22}$ are positive-definite and thus invertible.} ${\underline{\xx}^{*} = - A_{22}^{1} \mathbf a_{21} x_1 \eqqcolon B x_1}$. The quadratic form in fact reduces to ${\xx^\top A \xx = \left( a_{11} - \mathbf a_{21}^\top A_{22}^{-1} a_{21} \right ) x_1^2 \coloneqq \nu\, x_1^2}$, where both constants
\begin{equation}\label{eq:ConstPDP}
\begin{array}{rl}
\nu &=\displaystyle \EE_\mu\left[ \left( \rho\sqrt{Y} + \dot{\varphi}(Y)g(Y)\right )^2 \right ] + \EE_\mu\left [\brho^2 Y\right ] - \frac{1}{2}\EE_\mu\left[ \rho\sqrt{Y} + \dot{\varphi}(Y)g(Y)\right]^2 - \frac{\EE_\mu[\brho\sqrt{Y}]^2}{2},\\
B &=\displaystyle -\frac{1}{2}\begin{pmatrix}
\EE_\mu[\rho\sqrt{Y}+\dot{\varphi}(Y)g(Y)]\\
\EE_\mu[\brho\sqrt{Y}],
\end{pmatrix}
\end{array}
\end{equation}
are obtained from the definition of $Q$. Finally, we have
\begin{equation}\label{pb-deterministic-large}
\begin{cases}
&\displaystyle x_1^{*}=\argmax\limits_{x_1 \in L^2([0,T],\R)}
\left\{F\left(\int_{0}^{\cdot} x_1\D t\right) - \frac{\nu}{4}\int_0^T x_1^2 \D t\right\},
\\
&\dot{\hh}^{*}
= B x_1^{*}.
\end{cases}
\end{equation}
\begin{remark}
This problem is similar to the problem in Section~\ref{sec:BS}, with deterministic volatility.
\end{remark}

\begin{example}(Heston model) We again consider the Heston model, i.e., in the setting of~\eqref{eq:log_price_dyn} $f(v)=\kappa(\theta - v)$ and $g(v)=\xi\sqrt{v}$ for $\kappa,\theta>0$. 
Notice that the condition of Assumption~\ref{ass:SDE_V} are automatically satisfied. 
In this case, the invariant measure~$\mu$ is a Gamma distribution $\Gamma\big(\frac{2\kappa\theta}{\xi^2}, \frac{2\kappa}{\xi^2}\big)$ 
as shown in~\cite{cox1985theory}. 
Therefore, in light of Theorem~\ref{thm:limitPoisson}, 
$\overline{v} = \theta$ and the solution of the Poisson equation $\dot{\varphi}(y)=-\frac{1}{2\kappa}$ is constant, which means
that
$\varphi(y) = (\theta-y)/(2\kappa)$
because of the constraint $\int_{0}^{\infty}\varphi(y)\mu(\D y)=0$.
Therefore 
$$
\EE_\mu[Y] = \theta
\qquad\text{and}\qquad
\EE_\mu\left[\sqrt{Y}\right]
=\frac{\Gamma\left( \frac{2\kappa}{\xi^2}\theta + \frac{1}{2} \right )}{\Gamma\left( \frac{2\kappa\theta}{\xi^2} \right )} \frac{\xi}{\sqrt{2\kappa}}.
$$
From this, noting that $\dot{\varphi}(y)g(y) = -\xi\sqrt{y}/(2\kappa)$ we can calculate the constants~\eqref{eq:ConstPDP}:
\begin{align*}
\nu = \left( \left( \rho - \frac{\xi}{2\kappa} \right)^2 + 1-\rho^2 \right)\left( \theta - \EE_\mu\left[\sqrt{Y}\right]\right)
\qquad \text{and} \qquad
B = -\frac{1}{2}
\begin{pmatrix}
\rho - \frac{\xi}{2\kappa}\\
\brho\\
\end{pmatrix}\EE_\mu\left[\sqrt{Y}\right].
\end{align*}
\end{example}

\section{Numerical results}\label{sec:Num_results}
In all the different settings we studied, the final form of the optimisation problem is
\[\sup_{(\varphi, \phi)\in \HH_T^0\times \HH_T^{v_0}}F(\varphi) - \frac{1}{2}\int_0^T\ell(\varphi_t,\phi_t)\D t\]
where the function $F$ was linked to the payoff, $\ell$ to the rate function of the (log\nobreakdash
-)price process and $\varphi$ and $\phi$ are absolutely continuous paths that arise from Varadhan's lemma of the (log\nobreakdash-)price and volatility processes respectively.

In the tables below, we summarise all problems considered so far. As one can see, in the deterministic drift setting, the MDP problem is usually as simple as solving the problem under the Balck-Scholes (BS) model or at least by approximating the model with a Black-Scholes model. Furthermore, the variance reduction for geometric Asian options are also similar under the MDP with deterministic drift change, meaning advantage over a simple BS model is not significant. However, when it comes to the stochastic change of drift of the form $\int_{0}^{\cdot}\dot{h}_t\sqrt{V_t}\D t$, the MDP problems are slightly harder than in BS approximation and the variance reduction results are in fact significantly better.

\begin{table}[H]
\centering
 \makebox[\textwidth]{\begin{tabular}{ | m{4cm} | m{5.75cm}| m{5.5cm} | } 
 \hline
 \textbf{Method used} &  \textbf{Optimisation with general payoff} &  \textbf{Optimisation with payoff $\alpha\cT X$}\\
\hline
Deterministic volatility approximation (BS) & Optimisation in $\HH_T^0$ with $\ell$ a quadratic function of  the (log-)price process $\varphi$ & Optimisation on $\R$ with simple function to optimise \\
\hline
LDP small-(noise/time) & Optimisation involves solving a system of ODEs & Optimisation in $\RR^2$ with a time-consuming function to compute (the solution of a Ricatti)\\
\hline
MDP small-noise log-price & Optimisation has a simpler form than in the LDP. The ODE can be pre-computed & Straightforward optimisation in $\RR^2$ \\
\hline
MDP small-noise price & Optimisation has a simpler form than in the LDP. The ODE can be pre-computed & Straightforward optimisation in $\RR$ \\
\hline
MDP small-time & Similar complexity to the BS case & Similar complexity to the BS case \\
\hline
MDP large-time 
& Similar complexity to the BS case
& Similar complexity to the BS case
\\
\hline
\end{tabular}}
\caption{Summary of the optimisation problems with the \textit{deterministic} change of drift. $\varphi$ and $\phi$ are absolutely continuous paths that arise from Varadhan's lemma of the (log-)price and volatility processes respectively.}
\end{table}
\begin{table}[H]
\centering
 \makebox[\textwidth]{\begin{tabular}{ | m{4cm} | m{5.75cm}| m{5.5cm} | } 
 \hline
 \textbf{Method used} &  \textbf{Optimisation with general payoff} &  \textbf{Optimisation with payoff $\alpha\cT X$} \\
\hline
Deterministic volatility approximation (BS) & Optimisation on $\HH_T^0$ with $\ell$ a quadratic function of $\varphi$& Optimisation on $\R$ with simple function to optimise \\
\hline
LDP small-(noise/time) & Same as in the deterministic case with a simple additional ODE & Optimisation in $\RR^2$ with a time-consuming function to compute (the solution of an ODE) \\
\hline
MDP small-(noise/time) & Similar as in the deterministic case with an additional ODE & Straightforward optimisation in $\RR^2$ \\
\hline
\end{tabular}}
\caption{Summary of the optimisation problems with the \textit{stochastic} change of drift. $\varphi$ and $\phi$ are absolutely continuous paths that arise from Varadhan's lemma of the (log-)price and volatility processes respectively.}
\end{table}

\subsection{Pricing Asian options}
In order to compare variance reduction results in the Heston model, we look at Asian Geometric Call options, with payoffs of the form
\[
\left(\exp\left\{\frac{1}{T}\int_0^T X_t\D t\right\} - K\right)^+ = \left(S_0 \exp\left \{\frac{rT}{2}\int_0^T\frac{T-t}{T}\D X_t\right \} - K\right)^+,
\]
where $x^+ = \max\{x,0\}$. For convenience we restate the dynamics of the Heston model
\begin{equation*}
\begin{array}{rlrl}
\D X_t &= -\frac{1}{2}V_t\D t + \sqrt{V_t}(\rho \D W_t + \brho\D \Wp_t), & X_0 &= 0,\\
\D V_t &= \kappa (\theta - V_t)\D t + \xi \sqrt{V_t}\D W_t, & V_0 &= v_0 > 0,
\end{array}
\end{equation*}
with parameters realistic on Equity markets:
\[
S_0 = 50 \ ; \ 
r = 0.05 \ ; \ 
v_0 = 0.04 \ ; \ 
\rho = -0.5 \ ; \ 
\kappa = 2 \ ; \ 
\theta = 0.09 \ ; \ 
\xi = 0.2 \,.
\]
To simulate the paths $(X,V)$ on $[0,T]$, 
we use a standard Euler-Maruyama scheme for~$X$, but use the scheme~\cite{Lord2009} for the CIR process in the volatility, 
which is upward biased\footnote{There are many discretisation schemes for the Heston model. 
Since the objective of this paper is not to study the effects of different schemes we satisfy ourselves with~\cite{Lord2009}.}, however nevertheless converges strongly in $L^1$ to the true process $V$. For $n\in\NN$, $\Delta =\frac{T}{n}$ and the increments of the Brownian motion $\{\Delta W_{i}^n\}_{i=0}^{n-1}$ the scheme reads on $[0,T]$:
\begin{equation*}
\left\{
\begin{array}{rcl}
\widetilde{V}_{0}^n & = & v_0,\\
\widetilde{V}_{i+1}^n & = & \widetilde{V}_{i}^n + \kappa \left(\theta - \widetilde{V}_{i}^{n,+}\right)\Delta + \xi \sqrt{\widetilde{V}_{i}^{n,+}}\Delta W_{i}^n \quad \text{for all} \quad i\in\{0,\dots,n-1\},\\
V_{i}^n & = & \widetilde{V}_{i}^{n,+},
\end{array}
\right.
\end{equation*}
In what follows, we compare different LDP and MDP results, with $n=252$ trading days per year. 
All the results are computed for maturity~$T=1$, using $N_{\text{MC}}=500,000$ Monte-Carlo samples. We also consider an antithetic estimator and an LDP estimator derived under the assumption of deterministic volatility (denoted by BS).
Furthermore, since LDP-based deterministic changes of drift in the BS setting (or in cases where the final form of the optimisation problem is similar) are easy to compute, we also propose a \textit{fully adaptive} scheme based on the BS estimator: 
$$h_t=\sum_{i=1}^{n} h^i_t \, \ind_{[(i-1)\cdot\Delta, i\cdot\Delta)}(t),$$
where $h^i_t$ is the best deterministic change of drift up to the $i$-th discretisation step.\footnote{Fully adaptive schemes are  computationally very heavy, therefore we only consider it in the Black-Scholes setting. 
The change of law is computed $n \times N_{\text{MC}}$ times ($252 \times 5\mathrm{E}5 = 1.26\mathrm{E}8$ in our case).} 
We shall refer to \textit{deterministic} schemes
to mean changes of law with deterministic changes of drift and to \textit{adaptive} changes of drift 
for changes of law with drift of the form $\int_{0}^{\cdot} \dot{h}_t\sqrt{V_t}\D t$.

\subsection{LDP results in different settings}
We now look at the results of LDP based estimators in \textit{small-noise}, \textit{small-time} and \textit{large-time} setting. Figure~\ref{fig:LDPvarianceRedu} indicates that the estimators derived in small-noise log-price and small-noise price have very similar variance reduction. On the other hand, the small-time estimator provides good results, but is significantly outperformed by the other two. Although not apparent in the figure, looking at Table~\ref{tab:VarRedu}, the adaptive estimators provide slightly better results, as a matter of fact, they are notably better for small strikes. However, the computation time is also higher for adaptive estimators, which balances out the slight increase in variance reduction for higher strikes.
\begin{figure}[H]
    \centering
    \includegraphics[width=0.495\textwidth, trim={2cm 0.5cm 2cm 1cm }]{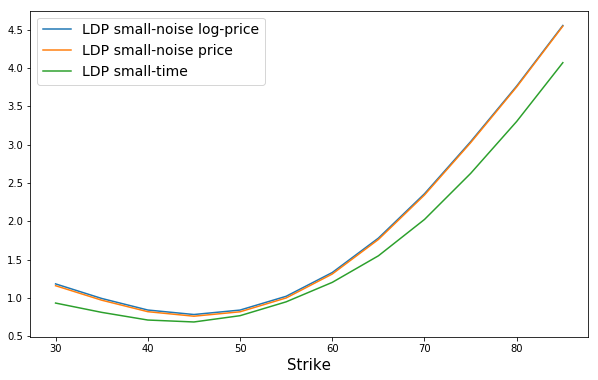}
    \includegraphics[width=0.495\textwidth, trim={2cm 0.5cm 2cm 1cm }]{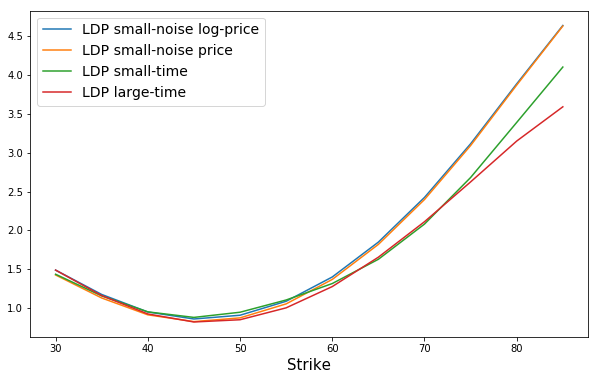}
    \caption{Variance reduction for LDP based estimators for in log-scale. \textit{Left}: deterministic change of drift. \textit{Right}: adaptive changes of drift.}\label{fig:LDPvarianceRedu}
\end{figure}

\subsection{MDP results in different settings}

In the deterministic case, all considered estimators have similar variance reduction (see Figure~\ref{fig:MDPvarianceRedu}). 
To be more precise, the BS estimator has a very similar variance reduction or even even slightly outperforms the MDP based estimators (note that the blue and green lines in the left plot of Figure~\ref{fig:MDPvarianceRedu} are indistinguishable for high strikes). Therefore, in that aspect, MDP based estimators do not justify their higher computational cost compared to the simple LDP-BS estimator. In the adaptive case, the BS estimator performs slightly better than before, whereas the MDP based estimators significantly outperform their results from the deterministic case \textit{and} those of the BS estimator. Moreover, as it will be discussed in the next section, their variance reduction is in fact even close to that of the LDP based estimators.
\begin{figure}[H]
    \centering
    \includegraphics[width=0.495\textwidth, trim={2cm 0.5cm 2cm 1cm }]{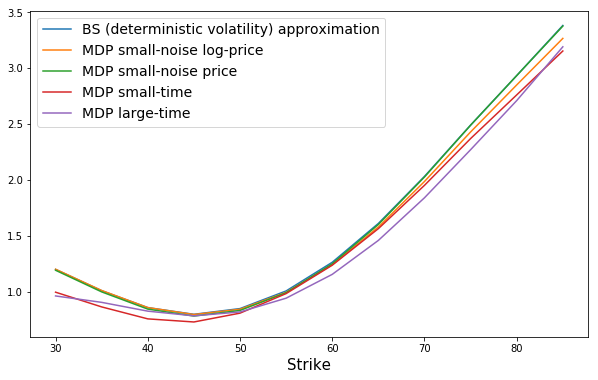}
    \includegraphics[width=0.495\textwidth, trim={2cm 0.5cm 2cm 1cm }]{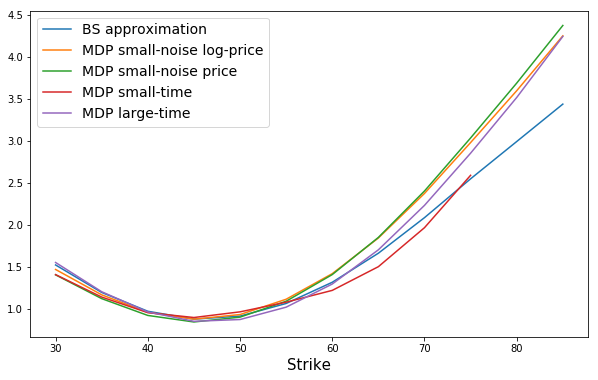}
    \caption{Variance reduction for MDP based estimators for in log-scale. \textit{Left}: deterministic change of drift. \textit{Right}: adaptive changes of drift. Note that because of computational problems, the adaptive small-time MDP estimator was not computed correctly for a strikes greater than $75$. Nevertheless, it looks to be outperformed significantly by other MDP estimators.}\label{fig:MDPvarianceRedu}
\end{figure}

\subsection{Overall comparison}

Looking at Figure~\ref{ref:OverallVarRedu}, as expected the LDP small-noise adaptive estimators perform best, even though MDP small-noise and large-time adaptive estimators are not far behind. 
Regaring computation time, Table~\ref{tab:CompTime} indicates that MDP estimators are on average about 
$10\%$  and LDP estimators approximately $15\%$ slower than the corresponding standard BS estimators. 
The fully adaptive BS estimator provides interesting results, especially for near-the-money strikes, 
where it performs much better than MDP and LDP estimators. 
Although this estimator is time consuming, it can still provide a good balance between variance reduction and computation time for certain strikes, see Tables~\ref{tab:VarRedu},~\ref{tab:VarComp},~\ref{tab:CompTime} and Figure~\ref{ref:OverallVarComp}.

\begin{figure}[hbt!]
    \centering
    \includegraphics[width=0.495\textwidth, trim={2cm 0.5cm 2cm 1cm }]{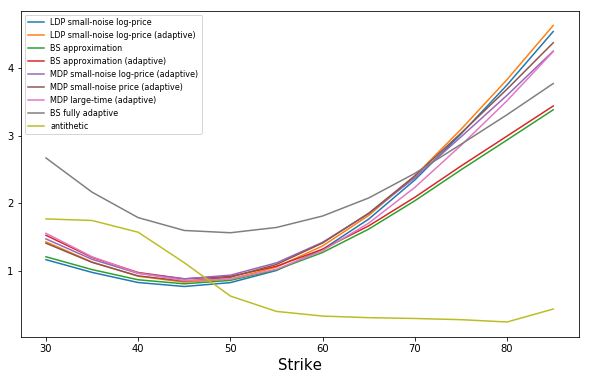}
    \includegraphics[width=0.495\textwidth, trim={2cm 0.5cm 2cm 1cm }]{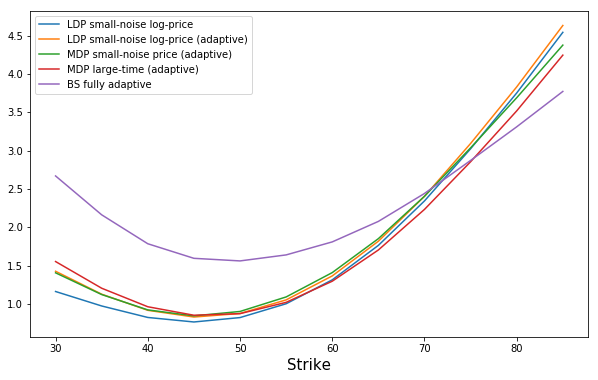}
    \caption{Variance reduction for different estimators in log-scale. 
    The antithetic estimator offers almost no variance reduction for OTM options, 
    because with higher strikes very few paths end up in the money, reducing the effect of antithetic samples.} \label{ref:OverallVarRedu}
\end{figure}

\begin{figure}
    \centering
    \includegraphics[scale=0.5]{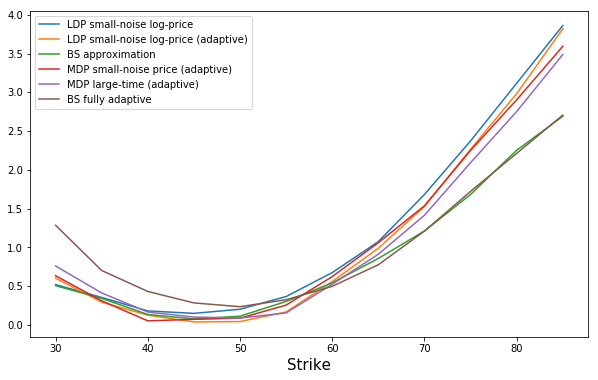}
    \caption{Ratio of variance reduction over computation time for different estimators in log-scale.}\label{ref:OverallVarComp}
\end{figure}
\noindent In the following three tables we use these notations:
\begin{enumerate}[-]
    \item \textbf{Proba}: Probability of having a positive Payoff.
    \item \textbf{LDPsn}: Deterministic estimator based on LDP in small-noise log-price setting.
    \item \textbf{LDPsn} A: Adaptive estimator based on LDP in small-noise log-price setting.
    \item \textbf{BS}: Deterministic BS estimator.
    \item \textbf{BS A}: Adaptive BS estimator.
    \item \textbf{MDPsn-${\log}$ A}: Adaptive estimator based on MDP in small-noise log-price setting.
    \item \textbf{MDPsn A}: Adaptive estimator based on MDP in small-noise price setting.
    \item \textbf{BS A2}: Fully adaptive BS estimator.
    \item \textbf{Ant}: Antithetic estimator.
    \item \textbf{Classic}: Classic Monte-Carlo estimator.
\end{enumerate}

\begin{table}[hbt!]
\centering
\begin{tabular}{lrrrrrrrHrr}
\toprule
Strike &   Prob. &  LDPsn &  LDPsn A &   BS &  BS A &  MDPsn$_{\log}$ A &  MDPsn A &  MDPlt A &  BS A2 &  Ant \\
\midrule
30 &    0.95 &     14 &       26 &   16 &    33 &           29 &       25 &       35 &    {470} &   58 \\
35 &    0.94 &    9.4 &       13 &   10 &    15 &           14 &       13 &       15 &    {150} &   55 \\
40 &     0.9 &    6.6 &      8.2 &  7.3 &   9.3 &          9.1 &      8.3 &      9.2 &     {60} &   36 \\
45 &    0.76 &    5.8 &      6.7 &  6.4 &   7.5 &          7.5 &      7.0 &      7.1 &     {39} &   13 \\
50 &    0.52 &    6.6 &      7.5 &  7.1 &   8.2 &          8.5 &      8.0 &      7.5 &     {36} &  4.2 \\
55 &    0.26 &     10 &       11 &   10 &    11 &           13 &       12 &       10 &     {43} &  2.5 \\
60 &   0.096 &     20 &       23 &   18 &    20 &           26 &       25 &       19 &     {64} &  2.1 \\
65 &   0.025 &     58 &       65 &   41 &    46 &           69 &       70 &       50 &    {120} &  2.0 \\
70 &   0.005 &    220 &      250 &  110 &   120 &          240 &      250 &      170 &    {280} &  1.9 \\
75 & 0.00078 &   1100 &     {1200} &  310 &   350 &          960 &     1100 &      720 &    750 &  1.9 \\
80 &  0.0001 &   5700 &     {6800} &  860 &   990 &         4000 &     4900 &     3300 &   2000 &  1.7 \\
85 & 1.1e-05 &  35000 &    {43000} & 2400 &  2800 &        18000 &    24000 &    18000 &   5900 &  2.7 \\
\bottomrule
\end{tabular}
\caption{Variance reduction for different estimators and probability of positive payoff}\label{tab:VarRedu}
\end{table}

\begin{table}[hbt!]
\centering
\begin{tabular}{lrrrrrrrHrr}
\toprule
Strike &   Prob. &  LDPsn &  LDPsn A &  BS &  BS A &  MDPsn$_{\log}$ A &  MDPsn A &  MDPlt A &  BS A2 &  Ant \\
\midrule
30 &    0.95 &    3.3 &      4.0 & 3.2 &   5.6 &          4.6 &      4.3 &      5.8 &     {19} &  8.2 \\
35 &    0.94 &    2.3 &      2.0 & 2.2 &   2.9 &          2.4 &      2.1 &      2.6 &    5.0 &  {7.4} \\
40 &     0.9 &    1.5 &      1.3 & 1.4 &   1.5 &          1.4 &      1.1 &      1.5 &    2.7 &  {5.9} \\
45 &    0.76 &    1.4 &      1.1 & 1.2 &   1.4 &          1.3 &      1.2 &      1.3 &    {1.9} &  1.8 \\
50 &    0.52 &    1.6 &      1.1 & 1.3 &   1.2 &          1.2 &      1.2 &      1.2 &    {1.7} & 0.53 \\
55 &    0.26 &    {2.3} &      1.5 & 2.0 &   1.9 &          1.8 &      1.8 &      1.4 &    2.1 &  0.4 \\
60 &   0.096 &    {4.7} &      3.6 & 3.5 &   3.7 &          4.2 &      4.2 &      3.4 &    3.2 & 0.34 \\
65 &   0.025 &     {11} &      9.8 & 7.2 &   7.4 &           {11} &       {11} &      8.2 &    6.0 & 0.32 \\
70 &   0.005 &     {48} &       33 &  16 &    21 &           31 &       34 &       26 &     16 & 0.31 \\
75 & 0.00078 &    {240} &      180 &  48 &    63 &          160 &      180 &      120 &     52 & 0.29 \\
80 &  0.0001 &   {1300} &      960 & 180 &   180 &          660 &      800 &      570 &    160 & 0.28 \\
85 & 1.1e-05 &   {7300} &     6600 & 490 &   500 &         3000 &     3900 &     3100 &    510 & 0.43 \\
\bottomrule
\end{tabular}
\caption{Ratio of variance reduction over computation time for different estimators}\label{tab:VarComp}
\end{table}

\begin{table}[hbt!]
\centering
\begin{tabular}{lrrrrrrrHrr}
\toprule
Strike &  Classic &  LDPsn &  LDPsn A &  BS &  BS A &  MDPsn$_{\log}$ A &  MDPsn A &  MDPlt A &  BS A2 &  Ant \\
\midrule
30 &       12 &     12 &       14 &  12 &    13 &           14 &       13 &       13 &     31 &   14 \\
35 &       11 &     11 &       14 &  12 &    13 &           13 &       14 &       13 &     36 &   15 \\
40 &       11 &     12 &       13 &  13 &    14 &           14 &       15 &       13 &     30 &   13 \\
45 &       11 &     11 &       13 &  13 &    13 &           13 &       13 &       13 &     28 &   14 \\
50 &       11 &     11 &       14 &  13 &    14 &           15 &       14 &       13 &     28 &   15 \\
55 &       11 &     12 &       15 &  12 &    13 &           15 &       14 &       15 &     28 &   13 \\
60 &       11 &     12 &       14 &  12 &    13 &           13 &       13 &       13 &     28 &   13 \\
65 &       11 &     12 &       14 &  13 &    13 &           13 &       13 &       13 &     27 &   13 \\
70 &       11 &     12 &       15 &  14 &    13 &           15 &       15 &       14 &     24 &   13 \\
75 &       11 &     12 &       14 &  14 &    13 &           13 &       13 &       13 &     21 &   14 \\
80 &       11 &     12 &       14 &  12 &    13 &           13 &       13 &       13 &     20 &   13 \\
85 &       11 &     12 &       14 &  12 &    13 &           13 &       13 &       13 &     19 &   13 \\
\bottomrule
\end{tabular}
\caption{Computation time (in seconds) for different estimators.}\label{tab:CompTime}
\end{table}

\newpage
\subsection{Variance swaps}
The methodology can also be applied to options with payoffs depending on volatility,
for example for options with payoffs of the form
$$
\int_0^T V_t \ind_{\{S_t\geq K\}}\D t.
$$
Consider $T = 1$ and different strikes $K>0$ under the Heston model with the same parameters as above. The results for different estimators are summarised in Table~\ref{tab:VarSwap}.
\begin{table}[H]
\centering
\begin{tabular}{lrrrrrrr}
\toprule
Strike &  LDPsn &  LDPsn A &  MDPsn &  MDPsn A &  BS &  BS A &  Ant \\
\midrule
10  &    {220} &       43 &    190 &       44 & 1.0 &   1.0 &   45 \\
20  &    {160} &       38 &    140 &       38 & 1.0 &   1.0 &   46 \\
30  &    8.2 &      6.3 &    8.6 &      6.5 & 1.0 &   1.0 &   {14} \\
40  &    1.1 &      1.1 &    1.1 &      1.1 &   1 &     1 &  {2.5} \\
45  &   0.86 &     0.85 &   0.88 &     0.87 & 1.0 &   1.0 &  {2.9} \\
50  &   0.96 &     0.96 &   0.86 &     0.76 & 1.0 &   1.0 &   {12} \\
55  &    2.4 &      2.5 &    2.6 &      2.4 & 1.8 &   1.9 &  {6.8} \\
60  &    3.1 &      3.3 &    {4.4} &      3.0 & 3.8 &   4.1 &  3.2 \\
70  &    8.4 &  {9.1} &    7.4 &      6.9 & 6.4 &   7.0 &  2.1 \\
80  &     16 &  19 &     15 &       {22} &  11 &    11 &  2.0 \\
90  &     56 &  {68} &     26 &       54 &  17 &    19 &  1.9 \\
100 &    200 &      {240} &     36 &      190 &  42 &    54 &  2.3 \\
\bottomrule
\end{tabular}
\caption{Variance reduction for different estimators.}\label{tab:VarSwap}
\end{table}

For small strikes, the LDP estimators based on BS approximation are not performing well, which is not surprising, since in the BS approximation the payoff of the option in question is almost constant for small strikes. On the other hand, LDP and MDP give good results. We also notice a clear difference in favor of non-adaptive changes of drift.
For large strikes, we have the same behavior as before: adaptive MDP estimators give intermediate results between LDP (performing best) and Black-Scholes.

\bibliography{biblio} 

\begin{thebibliography}{10}

\bibitem{Baldi2011}
{\sc P.~Baldi and L.~Caramellino}, {\em General {F}reidlin-{W}entzell large
  deviations and positive diffusions}, Statistics {\&} Probability Letters, 81
  (2011), pp.~1218--1229.

\bibitem{binder2012monte}
{\sc K.~Binder, D.~M. Ceperley, J.-P. Hansen, M.~Kalos, D.~Landau, D.~Levesque,
  H.~Mueller-Krumbhaar, D.~Stauffer, and J.-J. Weis}, {\em Monte Carlo methods
  in statistical physics}, vol.~7, Springer Science \& Business Media, 2012.

\bibitem{Chiarini2014}
{\sc A.~Chiarini and M.~Fischer}, {\em On large deviations for small-noise
  {I}t{\^{o}} processes}, Advances in Applied Probability, 46 (2014),
  pp.~1126--1147.

\bibitem{conforti2015small}
{\sc G.~Conforti, S.~De~Marco, and J.-D. Deuschel}, {\em On small-noise
  equations with degenerate limiting system arising from volatility models}, in
  Large Deviations and Asymptotic Methods in Finance, Springer, 2015,
  pp.~473--505.

\bibitem{cox1985theory}
{\sc J.~Cox, J.~Ingersoll, and S.~Ross}, {\em A theory of the term structure of
  interest rates}, Econometrica, 53 (1985), pp.~385--407.

\bibitem{Dembo2010}
{\sc A.~Dembo and O.~Zeitouni}, {\em Large deviations techniques and
  applications}, Springer, 2010.

\bibitem{Donati-Martin2004}
{\sc C.~Donati-Martin, A.~Rouault, M.~Yor, and M.~Zani}, {\em Large deviations
  for squares of {B}essel and {O}rnstein-{U}hlenbeck processes}, Probability
  Theory and Related Fields, 129 (2004), pp.~261--289.

\bibitem{Dupuis2017}
{\sc P.~Dupuis and D.~Johnson}, {\em Moderate deviations-based importance
  sampling for stochastic recursive equations}, Advances in Applied
  Probability, 49 (2017), pp.~981--1010.

\bibitem{Dupuis2012}
{\sc P.~Dupuis, K.~Spiliopoulos, and H.~Wang}, {\em Importance sampling for
  multiscale diffusions}, Multiscale Modeling \& Simulation, 10 (2012),
  pp.~1--27.

\bibitem{dupuis2004importance}
{\sc P.~Dupuis and H.~Wang}, {\em Importance sampling, large deviations, and
  differential games}, Stochastics: An International Journal of Probability and
  Stochastic Processes, 76 (2004), pp.~481--508.

\bibitem{Freidlin2012}
{\sc M.~I. Freidlin and A.~D. Wentzell}, {\em Random perturbations of dynamical
  systems}, Springer, 2012.

\bibitem{glasserman2004monte}
{\sc P.~Glasserman}, {\em Monte Carlo methods in {F}inancial {E}ngineering},
  vol.~53, Springer, 2004.

\bibitem{Glasserman1997}
{\sc P.~Glasserman and Y.~Wang}, {\em Counterexamples in importance sampling
  for large deviations probabilities}, The Annals of Applied Probability, 7
  (1997), pp.~731--746.

\bibitem{Grbac2021}
{\sc Z.~Grbac, D.~Krief, and P.~Tankov}, {\em Long-time trajectorial large
  deviations and importance sampling for affine stochastic volatility models},
  Advances in Applied Probability, 53 (2021), pp.~220--250.

\bibitem{Guasoni2007}
{\sc P.~Guasoni and S.~Robertson}, {\em Optimal importance sampling with
  explicit formulas in~continuous time}, Finance and Stochastics, 12 (2007),
  pp.~1--19.

\bibitem{Hartmann2015}
{\sc C.~Hartmann, C.~Schütte, M.~Weber, and W.~Zhang}, {\em Importance
  sampling in path space for diffusion processes with slow-fast variables},
  Probability Theory and Related Fields, 170 (2018), pp.~177--228.

\bibitem{Jacquier2019}
{\sc A.~Jacquier and K.~Spiliopoulos}, {\em Pathwise moderate deviations for
  option pricing}, Mathematical Finance, 30 (2019), pp.~426--463.

\bibitem{Klenke2014}
{\sc A.~Klenke}, {\em Probability Theory}, Springer London, 2014.

\bibitem{Lord2009}
{\sc R.~Lord, R.~Koekkoek, and D.~Van~Dijk}, {\em A comparison of biased
  simulation schemes for stochastic volatility models}, Quantitative Finance,
  10 (2009), pp.~177--194.

\bibitem{manly2018randomization}
{\sc B.~F. Manly}, {\em Randomization, bootstrap and Monte Carlo methods in
  biology}, {C}hapman and {H}all, 2018.

\bibitem{Morse2017}
{\sc M.~R. Morse and K.~Spiliopoulos}, {\em Moderate deviations for systems of
  slow-fast diffusions}, Asymptotic Analysis, 105 (2017), pp.~97--135.

\bibitem{Robertson2010}
{\sc S.~Robertson}, {\em Sample path large deviations and optimal importance
  sampling for stochastic volatility models}, Stochastic Processes and their
  Applications, 120 (2010), pp.~66--83.

\bibitem{Schilder1966}
{\sc M.~Schilder}, {\em Asymptotic formulas for {W}iener integrals},
  Transactions of the AMS, 125 (1966), pp.~63--85.

\bibitem{Siegmund1976}
{\sc D.~Siegmund}, {\em Importance sampling in the {M}onte {C}arlo study of
  sequential tests}, The Annals of Statistics,  (1976), pp.~673--684.

\bibitem{Varadhan1967}
{\sc S.~Varadhan}, {\em Diffusion processes in a small time interval},
  Communications on Pure and Applied Mathematics, 20 (1967), pp.~659--685.

\bibitem{Yamada1971}
{\sc T.~Yamada and S.~Watanabe}, {\em On the uniqueness of solutions of
  stochastic differential equations}, Kyoto Journal of Mathematics, 11 (1971).

\end{thebibliography}
\bibliographystyle{siam}

\appendix

\section{Technical proofs}\label{apx:proofs}

\subsection{Proof of Lemma~\ref{lem:degenerate}}\label{sec:lem:degenerate_Proof}
Suppose the random variable $Z^\eps:\Omega\rightarrow \Xx$ maps to metrisable space~$\Xx$. Let then $(\Xx, \mathrm{d})$ be a metric space let $\beta > 1$. We first show that $Z^{\eps^\beta}$ is exponentially equivalent to $x_0\in \Xx$. For $\delta > 0$ define $\Gamma_{\delta} \coloneqq \{x \in \Xx : \, \mathrm{d}(x,x_0)>\delta \, \}$ and observe that $\{\omega \in \Omega : \, Z^{\eps^\beta}(\omega)\in\Gamma_\delta\}$ is measurable, since $\Gamma_\delta$ is an open set. Note also that~$\II$ achieves the infimum over a closed set, being a good rate function. It thus follows that $\inf_{x\in\overline{\Gamma}_\delta} \II(x) > 0$, since $\II(x)=0$ if and only if $x=x_0$. Recall also that $Z^\eps$ satisfies an LDP with the good rate function~$\II$, then $\limsup_{\eps\downarrow 0}\eps^\beta\log\PP\left[ Z^{\eps^\beta} \in \Gamma_\delta \right] \leq - \inf_{x\in\overline{\Gamma}_\delta}\II(x)$ and therefore
\begin{align*}
\limsup_{\eps\downarrow 0} \eps \log \PP \left[ Z^{\eps^\beta} \in \Gamma_\delta \right] &= \limsup_{\eps\downarrow 0} \eps^{1-\beta}\left(\eps^\beta \log \PP \left[ Z^{\eps^\beta} \in \Gamma_\delta \right]\right) \\
&\leq -\bigg\{ \inf_{x\in\overline{\Gamma}_\delta} I(x)\bigg\} \limsup_{\eps\downarrow 0}\eps^{1-\beta} = -\infty.
\end{align*}
Thus proving the exponential equivalence. Next, we show that $x_0\in\Xx$ satisfies an LDP with the good rate function~$\II$. Let $\Gamma\subset\Xx$ such that $x_0\notin\Gamma$, then $\limsup_{\eps\downarrow 0} \eps \log \PP[x_0\in\Gamma]=-\infty$ for the upper bound and thus $-\inf_{x\in\Gamma^\circ}=-\infty$. Now let $\Gamma_{x_0}\subset\Xx$ be such that $x_0\in\Xx$, then $-\inf_{x\in\Gamma_{x_0}} \II(x)=0$ for the upper bound and $\liminf_{\eps\downarrow 0}\eps \log \PP[x_0\in\Gamma_{x_0}]=0$ for the lower bound. By the exponential equivalence the stated result follows.


\subsection{Proof of Lemma~\ref{lem:expct_convg}}\label{sec:lem:expct_convg_Proof}
Let $Z^\eps:\Omega\rightarrow \Xx$ 
with~$\Xx$ a metrisable space and$(\Xx, \mathrm{d})$ be a metric space, and $\delta,M>0$. 
Define further the sets ${\Gamma_{\delta} \coloneqq \{x \in \Xx: \, \mathrm{d}(x,x_0)>\delta \, \}}$ and ${B_M \coloneqq \left \{ x\in\Xx: \, |x| \leq M \, \right \}}$. 
Since $Z^\eps$ is uniformly integrable, there exists $\delta>0$ such that
\[
\sup_{\eps>0}\EE\left[ \left\vert Z^\eps \right\vert \ind_{\left\{ Z^\eps \notin B_M \right\}} \right] < \delta,
\]
and hence 
$$
\sup_{\eps>0}\EE\left[ \left\vert Z^\eps \right\vert \right] \leq \sup_{\eps>0}\EE\left[ \left\vert Z^\eps \right\vert \ind_{\left\{ Z^\eps \in B_M \right\}} \right] + \sup_{\eps>0}\EE\left[ \left\vert Z^\eps \right\vert \ind_{\left\{ Z^\eps \notin B_M \right\}} \right]
\leq M + \delta.
$$
By Fatou's lemma we have
\[
|x_0| \leq \liminf_{\eps\downarrow 0} \EE\left[ \left\vert Z^\eps \right\vert \right] < \sup_{\eps>0} \left[ \left\vert Z^\eps \right\vert \right],
\]
but also by Bonferroni's inequality
\begin{align*}
\EE\left[ \left\vert Z^\eps \right\vert \right] &= \EE\left[ \left\vert Z^\eps \right\vert \ind_{\left \{ Z^\eps \in \Gamma_\delta \cap Z^\eps \notin B_M \right \}} \right] + \EE\left[ \left\vert Z^\eps \right\vert \ind_{\left \{ Z^\eps \in \Gamma_\delta \cap Z^\eps \in B_M \right \}} \right] + \EE\left[ \left\vert Z^\eps \right\vert \ind_{\left \{ Z^\eps \in \Gamma_\delta \right \}} \right].
\end{align*}
Since $ \left\vert Z^\eps \right\vert \ind_{\left \{ Z^\eps \notin \Gamma_\delta \right \}}\leq \left ( \mathrm{d}\left(Z^\eps, x_0\right) + |x_0| \right ) \ind_{\left \{ Z^\eps \notin \Gamma_\delta \right \}}$, then $\EE\left[ \left\vert Z^\eps \right\vert \right]\ind_{\left \{ Z^\eps \notin \Gamma_\delta \right \}} \leq \delta + |x_0|$. Finally
\begin{align*}
\limsup_{\eps\downarrow 0} \EE\left[ \left\vert Z^\eps \right\vert \right]  &\leq \limsup_{\eps\downarrow 0}\EE\left[ \left\vert Z^\eps \right\vert \ind_{\left \{Z^\eps \notin B_M \right \}} \right] + M\lim_{\eps\downarrow 0} \PP\left[ Z^\eps \in \Gamma_\delta \right ] + x_0.
\end{align*}
Because $\lim_{\eps\downarrow 0} \PP\left[ Z^\eps \in \Gamma_\delta \right ] = 0 $ from the proof of Lemma~\ref{lem:degenerate} and $Z^\eps$ is uniformly integrable, the result follows by taking $M$ to infinity.


\subsection{Proof of Lemma~\ref{lem:degenerate-limit}}\label{sec:lem:degenerate-limit_Proof}
A key step in proving Lemma~\ref{lem:degenerate-limit}
is the tightness of the rescaled variance process:
\begin{lemma}\label{lem:tightness}
The family of random variables $\{V^\eps\}_{\eps>0}$ from~\eqref{eq:log_price_dyn} is tight.
\end{lemma}

\begin{proof}[Proof of Lemma~\ref{lem:tightness}]
By Kolmogorov-Chentsov~\cite[Theorem 21.42]{Klenke2014} we need to show there exist $\alpha, \beta, M > 0$ such that for every $\eps >0$ and $0\leq s < t \leq T$
\[
\EE\left[\left\vert V_t^\eps - V_s^\eps\right\vert^\alpha\right] \leq M \left\vert t - s\right\vert^{1+\beta},
\]
therefore using the obvious inequality 
$|a+b|^{\alpha} \leq 2^{\alpha-1}\left(|a|^{\alpha}+|b|^{\alpha}\right)$ for $\alpha\geq1$ we have
\begin{align*}
\EE\left[ \left\vert \int_0^t f(V_u^\eps)\D u - \int_0^s f(V_u^\eps)\D u + \int_0^t g(V_u^\eps)\D u - \int_0^s g(V_u^\eps)\D u \right \vert^\alpha \right ] \\ \leq 2^{\alpha-1} \left( \EE\left[ \left\vert  \int_s^t f(V_u^\eps) \D u \right\vert^\alpha \right ] + \EE\left[ \left\vert  \int_s^t g(V_u ^\eps) \D u \right\vert^\alpha \right] \right).
\end{align*}
Consider for now only the drift term
$$
\EE\left[ \left\vert  \int_s^t f(V_u^\eps) \D u \right\vert^\alpha \right ]
\leq (t-s)^{\alpha-1}\EE\left[ \int_s^t \left\vert f(V_u^\eps)\right\vert^\alpha \D u  \right ]
\leq K_1^\alpha(t-s)^{\alpha-1}\int_s^t \EE\left[ 1 + \left\vert V_u^\eps\right\vert^\alpha \right]\D u,
$$
where we used Jensen's inequality for the first inequality. 
The second follows from the linear growth condition in Assumption~\ref{ass:SDE_coeffs}. For the diffusion term
\begin{align*}
\EE\left[ \left\vert  \int_s^t g(V_u ^\eps) \D u \right\vert^\alpha \right] 
&\leq \left\{\frac{\alpha(\alpha-1)}{2}\right\}^{\frac{\alpha}{2}} \EE\left[ \left\vert  \int_s^t \left\vert g(V_u ^\eps)\right\vert^{2} \D u \right\vert^{\frac{\alpha}{2}} \right]\\
&\leq \left\{\frac{\alpha(\alpha-1)}{2}\right\}^{\frac{\alpha}{2}} (t-s)^{\frac{\alpha}{2}-1} \EE\left[\int_s^t \left\vert g(V_u ^\eps)\right\vert^{\alpha} \D u  \right] & 
\\
&\leq K_2^\alpha \left\{\frac{\alpha(\alpha-1)}{2}\right\}^{\frac{\alpha}{2}} (t-s)^{\frac{\alpha}{2}-1} \int_s^t \EE\left[1 +  \left\vert V_u^\eps\right\vert^{\alpha H} \right] \D u,
\end{align*}
where the first line follows from the Burkholder-Davis-Gundy inequality (with $\alpha\geq 2$) and the last one from the 
$H$-polynomial growth ($H\geq \half$) condition on the diffusion. 
Adding both terms together and applying the Gronwall lemma to the integrands yields
\begin{align*}
\EE\left[\left\vert V_t^\eps - V_s^\eps\right\vert^\alpha\right] 
&\leq 2^{\alpha-1}\Big(K_1^\alpha(t-s)^{\alpha-1}\int_s^t \EE\left[ 1 + \left| V_u^\eps\right|^\alpha \right]\D u \\ 
& \quad \quad \quad + K_2^\alpha \left\{\frac{\alpha(\alpha-1)}{2}\right\}^{\frac{\alpha}{2}} (t-s)^{\frac{\alpha}{2}-1} \int_s^t \EE\left[1 +  \left| V_u^\eps\right|^{\alpha H} \right] \D u\Big)\\
&\leq 2^{\alpha-1}\Big(K_1^\alpha(t-s)^{\alpha-1}\int_s^t 2^{\frac{\alpha}{2}-1}(1 + v_0^\alpha)\E^{2\alpha(u-s)} \D u \\ & ~\quad \quad \:\; + K_2^\alpha \left\{\frac{\alpha(\alpha-1)}{2}\right\}^{\frac{\alpha}{2}} (t-s)^{\frac{\alpha}{2}-1} \int_s^t 2^{\frac{\alpha}{2}-1}(1 + v_0^\alpha)\E^{2\alpha(u-s)} \D u\Big)\\
&\leq 2^{\alpha-1} \Big(2^{\frac{\alpha}{2}-1} K_1^{\alpha} (1+v_0^\alpha)\E^{2\alpha H T} (t-s)^\alpha +  2^{\frac{\alpha H}{2}-1} K_2^{\alpha} (1+v_0^{\alpha H})\E^{2\alpha H T} (t-s)^{\frac{\alpha}{2}} \Big) \\
&\leq 2^{\alpha-1}\left(C_1 T^{\frac{\alpha}{2}} + C_2\right)(t-s)^{\frac{\alpha}{2}}
= M \left\vert t - s\right\vert^{\frac{\alpha}{2}},
\end{align*}
where the constants are 
$$
C_1\coloneqq 2^{\frac{\alpha}{2}-1} K_1^{\alpha} (1+v_0^\alpha)\E^{2\alpha T},
\quad
C_2\coloneqq 2^{\frac{\alpha H}{2}-1} K_2^{\alpha} (1+v_0^{\alpha H})\E^{2\alpha H T},
\quad
M\coloneqq 2^{\alpha-1}\max\{C_1 T^{\frac{\alpha}{2}}, \, C_2\}.
$$
Choosing any $\alpha\geq 2$ and $\beta = \frac{\alpha}{2}-1$ completes the proof.
\end{proof}

Similarly as in~\cite{Chiarini2014}, define the bounded map $\Phi_t\in\Cc_T$ for each $t\in[0,T]$ as 
$$
\Phi_t(\psi) = \left \vert \psi_t - v_0 - \int_0^t f(\psi_t)\D s \right \vert \land 1.
$$
It is also continuous. 
Indeed, let $\psi^n \rightarrow \psi$ in $\Cc_T$, then we have by the triangle inequality 
$$
\left \vert \Phi_t(\psi^n) - \Phi_t(\psi) \right \vert \leq  \left \vert \psi^n - \psi \right \vert + \int_0^t \left \vert f(\psi^n) - f(\psi) \right \vert\D s.
$$
Since $\Aa\coloneqq \{\psi^n\in\Cc_T: n\in\NN\}\cup \{\psi\}$ is a compact subset of $\Cc_T$ and $f$ is Lipschitz continuous on~$\Aa$ by Assumption~\ref{ass:SDE_coeffs}, there exists a Lipschitz constant $L$, such that
\[
\left\vert f(\varphi_t) - f(\phi_t) \right\vert \leq L \sup_{t\in[0,T]} |\varphi_t - \phi_t|,
\]
for all $t\in[0,T]$ and $\varphi, \phi\in\Aa$. 
Therefore~$\Phi_t$ is continuous since
$$
\left \vert \Phi_t(\psi^n) - \Phi_t(\psi) \right \vert \leq \sup_{t\in[0,T]} \left \vert\psi_t^n - \psi_t\right \vert + t L \sup_{t\in[0,T]} \left \vert\psi_t^n - \psi_t\right \vert.
$$
Now by Lemma~\ref{lem:tightness} $\{V^\eps\}_{\eps>0}$ is tight as a family of random variables. Therefore, by taking a subsequence $\{V^\eps\}_{\eps>0}\,$, converges in distribution to some random variable on the same probability space. 
Since~$\Phi_t$ is continuous and bounded, we have by the Continuous mapping theorem that $\lim_{\eps\downarrow 0} \EE[\Phi_t(V^\eps)]=\EE[\Phi(\psi)]$. We now show this limit is zero for 
$\widetilde{\psi}=v_0 + \int_{0}^{\cdot} f(\widetilde{\psi}_s)\D s$ by H\"older inequality and It\^o isometry
\begin{align*}
\EE[\Phi_t(V^\eps)]
 = \EE\left[\left\vert \sqrt{\eps} \int_0^t g(V_s^\eps) \D W_s \right\vert\right]
  & \leq \left( \eps \; \EE\left[\left\vert \int_0^t g(V_s^\eps) \D W_s \right\vert^2\right] \right)^{\frac{1}{2}}\\
  &= \left( \eps \; \EE\left[ \int_0^t \left\vert g(V_s^\eps)\right\vert^2 \D W_s \right] \right)^{\frac{1}{2}}
  \leq \sqrt{\eps} M,
\end{align*}
since the integral is bounded by $M>0$ by Lemma~\ref{lem:tightness}.
Therefore $\lim_{\eps\downarrow 0} \EE[\Phi_t(V^\eps)] = 0$
and the limiting function indeed solves the ODE $\widetilde{\psi}=v_0 + \int_{0}^{\cdot} f(\widetilde{\psi}_s)\D s$ 
almost surely for all $t\in[0,T]$ by the definition of~$\Phi_t$.


\subsection{Proof of Theorem~\ref{thm:mdp-small-noise}}\label{sec:thm:mdp-small-noise_Proof}
We first prove the following version of Gronwall's lemma:
\begin{lemma}\label{lem:Gronwall-bis}
Let $0\leq a < b$ and $\varphi:[a,b]\rightarrow \RR^+$ continuous with
${\varphi(t)\leq M + \int_a^t f(s) g(\varphi(s))\D s}$ on~$[a,b]$ for some $M>0$.
If $f:[a,b]\rightarrow \RR$ and $g:\RR^+\rightarrow\RR^+$ satisfy Assumption~\ref{ass:SDE_coeffs}
(in particular, $f$ continuous and $g$ increasing, continuous and strictly positive outside the origin), then
$G(\varphi(t)) \leq G(M) + \int_a^t f(s)\D s$ for all $t \in [a,b]$,
where $G(u) = \int_1^u g(s)^{-1}\D s$, with $G(0) = -\infty$.
\end{lemma}

\begin{proof}
Let $\psi(t) := \int_a^{t} g(\varphi(s))f(s)\D s$, so that
$\dot{\psi}(t) = g(\varphi(t))f(t) \leq g(M+\psi(t)) f(t)$ on $[a,b]$,
since~$g$ is increasing. By integration,
$$
\int_a^tf(s)\D s \geq \int_a^t \frac{\dot{\psi}(s)}{g(M+\psi(s))}\D s 
= \int_{M+\psi(a)}^{M+\psi(t)}\frac{\D u}{g(u)}
=  \int_{M}^{\psi(t)+M}\frac{\D u}{g(u)},
$$
with the substitution $u:=M+\psi(s)$, since $\psi(a)=0$ by definition.
Therefore, since $G$ is increasing, then, for any $t\in[a,b]$,
$$
G(\varphi(t)) \leq G(\psi(t) + M) 
\leq \int_1^{M +\psi(z)}\frac{\D s}{g(s)}
\leq \left(\int_{1}^{M}+ \int_{M}^{M+\psi(t)}\right)\frac{\D s}{g(s)}
\leq G(M) + \int_a^t f(s)\D s.
$$
\end{proof}
Let $Y^\eps=\int_{0}^{\cdot} g(V^\eps_t)\D W_t^\eps$.
First, note that by Lemma~\ref{lem:degenerate} and Lemma~\ref{lem:degenerate-limit} the process~$V^\eps$ is exponentially equivalent to $\psi$ (with speed $\eps$),
so that $\{V^\eps,W^\eps\}$ is exponentially equivalent to $\{\psi,W^\eps\}$. 
Therefore, by a small extension to Theorem~\ref{thm:robertsonLDP} shown in~\cite[Lemma 3.1]{Robertson2010}, the triple $\{V^\eps,Y^\eps,W^\eps\}$ satisfies an LDP with the good rate function
$$
\II^{V,Y,W}(v,y,w) = 
\frac{1}{2}\int_0^T\dot{w}_t^2\D t,\quad \text{if }\ w \in \HH_T^0,\  \dot{Y} = g(\psi)\dot{w}, \ v = \psi,
$$
and is infinite otherwise. Now let 
\begin{equation}\label{eq:etatilde}
\widetilde{\eta}_\eps \coloneqq \eta^{\eps} - \int_{0}^{\cdot} f'(\psi_t)\eta^\eps(t) \D t = \int_{0}^{\cdot} \left\{f(V^\eps_t) - f(\psi_t) - f'(\psi_t)(V^\eps_t + \psi_t)\right\}\eps^{-\beta}\D t + Y^\eps.
\end{equation}
and suppose the first term is exponentially equivalent to zero. 
Then, by Contraction principle~\cite[Theorem 4.2.1]{Dembo2010} 
we deduce an LDP for $\{V^\eps, \widetilde{\eta}_\eps, W^\eps\}$. Moreover, since $\widetilde{\eta}^\eps \mapsto \eta^\eps$ is a continuous function\footnote{It is easy to show that $\eta^{\eps} =\E^{\int_{0}^{\cdot} f(s)\D s}\int_{0}^{\cdot} \E^{-\int_0^t f(s)\D s}\widetilde{\eta}_\eps(t) \D t $ solves the ODE in~\eqref{eq:etatilde}.}, the good rate function given in the statement is obtained using the Contraction principle once more. 
Therefore the rest of the proof relies on proving that 
\[
\left\{f(V^\eps) - \left(f(\psi) - f'(\psi)(V^\eps -\psi)\right)\right\}\eps^{-\beta}
\]
is exponentially equivalent to $0$. 
We start by showing $\eta^\eps$ is bounded. To that end we consider a Taylor expansion of $\ell(u) \coloneqq f(\psi+u(V^\eps - \psi))$ for $u\in\R$ around zero evaluated at $u=1$:
\[
\ell(1)=\ell(0)+\ell'(0) + \int_0^1 \ell''(u)(1-u)\D u,
\]
so that, since $\ell(0)=f(\psi)$ and $\ell(1)=f(V^\eps)$, we have 
\[
\left\{f(V^\eps) - f(\psi) - f'(\psi)(V^\eps -\psi)\right\}\eps^{-\beta} = \eps^\beta |\eta^\eps|^2 \int_0^1 (1-u)f''(\psi + u (V^\eps - \psi))\D u
\]
Now let $R>0$ and $\omega \in \Omega$ such that 
\[
\sup_{t \in [0,T]}|V_t^\eps(\omega) - \psi_t|\leq R
\qquad\text{and}\qquad 
\sup_{t \in [0,T]}|Y^\eps(\omega)| \leq \eps^{-\frac{\beta}{4}}.
\]
Then for all $t\in[0,T]$ we have
\begin{align*}
\left|\eta_{t}^\eps(\omega)\right| &= \left\lvert \int_0^t\left\{f(V^\eps_s(\omega)) - \left(f(\psi_s) - f'(\psi_s)(V_s^\eps(\omega) -\psi_s)\right)\right\}\eps^{-\beta}\D s + Y^{\eps}_t(\omega)\right\rvert \\
&= \left\lvert\int_0^t f'(\psi_s)\eta_{s}^\eps\D s + \int_0^t|\eta_{s}^\eps|^2\eps^{\beta} \int_0^1(1-u)f''(\psi_u + u(V^\eps_u(\omega) - \psi_u))\D u\,\D s + Y_t^\eps(\omega)\right\rvert \\
&\leq \int_0^t\big|f'(\psi_s)\big| |\eta_{s}^\eps|\D s + \int_0^t|\eta_{s}^\eps|^2\eps^{\beta} \int_0^1(1-u)\left|f''\left(\psi_u + u(V^\eps_u - \psi_u)\right)\right|\D u\,\D s + \left|Y_t^\eps(\omega)\right|\\
&\leq \alpha \int_0^t\left(|\eta_{s}^\eps| + |\eta_{s}^\eps|^2\eps^{\beta}\right) \D s + \eps^{-\frac{\beta}{4}},
\end{align*}
where 
\[
\alpha \coloneqq 1+\max\Big\{\sup_{x \in \mathcal{I}_\psi}\big|f'(x)\big| : \sup_{x\in \mathcal{I}\psi}\big|f''(x)\big| \Big\}>0 \quad \text{ and } \quad \mathcal{I}_\psi \coloneqq  \Big[-\sup_{u \in [0,1]}\left\{|\psi_u|+ R\right\},\sup_{u \in [0,1]}\left\{|\psi_u| + R\right\}\Big].
\]
Notice that $\alpha$ is finite, because continuous functions admit a maximiser on compact sets. Then by Lemma~\ref{lem:Gronwall-bis} $G(|\eta_{t}^\eps|) \leq G(\eps^{-\frac{\beta}{4}}) + \alpha t$ for all $t \in [0,T]$ and $u>0$ with 
$$
G(u) = \int_1^u\frac{1}{x + \eps^\beta x^2}\D x = \log\left(u\frac{1+\eps^\beta}{1 + \eps^\beta u}\right),
$$
we have
$$
G(|\eta_{t}^\eps|) 
\leq \log\left(\eps^{\frac{-\beta}{4}}\frac{1+\eps^\beta}{1 + \eps^{\frac{3}{4}\beta}}\right) + \alpha T,
\qquad\text{and}\qquad
|\eta_{t}^\eps| \leq \eps^{-\beta}\left\{\frac{1}{1 - C(\eps)\frac{\eps^\beta}{1 + \eps^\beta}} - 1 \right\},
$$
where 
$$
C(\eps) = \E^{\alpha T}\eps^{\frac{-\beta}{4}}\frac{(1+\eps^\beta)}{1 + \eps^{\frac{3}{4}\beta}}
 <\E^{\alpha T}\eps^{\frac{-\beta}{4}}, 
\quad \text{ for all } \quad 0<\eps<1.
$$
Therefore there exists $0<\eps_0<1$ such that
\[
\sup_{t \in [0,T]}|\eta_{t}^\eps(\omega)|\leq 2\E^{\alpha T}\eps^{-\frac{\beta}{4}},
\quad \text{ for all } \quad 0 < \eps <\eps_0.
\]
We now finally prove that $\eps^\beta |\eta^\eps|^2 \int_0^1 (1-u)f''(\psi + u (V^\eps - \psi))\D u$ is exponentially equivalent to zero.
Let $\delta >0$ and $0<\eps<\eps_0$, and define
\begin{align*}
&\Gamma_\eps := \Big\{\omega\in\Omega : \sup_{t \in [0,T]}\eps^\beta |\eta_{t}^\eps|^2 \Big|\int_0^1 (1-u)f''(\psi_t + u\ (V^{\eps}_{t} - \psi_t))\D u\Big|\geq \delta\Big\} \\
&A_\eps := \Big\{\omega\in\Omega : \sup_{t \in [0,T]}|V^\eps_t - \psi_t|\leq  R \Big\} \\
&B_\eps := \Big\{\omega\in\Omega : \sup_{t\in[0,T]}|Y_t^\eps|\leq \eps^{-\frac{\beta}{4}}\Big\},
\end{align*}
then
\begin{align*}
\eps\log\PP\Big[\Gamma_\eps\Big]
&\leq \eps\log \Big\{ \PP\Big[\Gamma_\eps \cap A_\eps \cap B_\eps\Big] + \PP\Big[A_\eps^c \Big] + \PP\Big[B_\eps^c   \Big]\Big\} \\
&\leq \eps\log3 + \max \Big \{ \eps\log \PP\Big[\Gamma_\eps \cap A_\eps \cap B_\eps\Big], \  \eps\log \PP\Big[A_\eps^c \Big], \  \eps\log \PP\Big[B_\eps^c \Big] \Big\}.
\end{align*}
Using the previous bound on $\eta^\eps$, which holds under condition $\{A_\eps\cap B_\eps\}$, we have that
\begin{align*}
\eps\log \PP\Big[\Gamma_\eps \cap A_\eps \cap B_\eps\Big]
&\leq \eps\log \PP\Big[\Big\{\omega\in\Omega : \sup_{t \in [0,T]}\eps^\beta |\eta_{t}^\eps|^2\geq \frac{1}{\alpha}\delta\Big\} \cap A_\eps \cap B_\eps\Big] \\
&\leq \eps\log \PP\Big[\eps^{\frac{\beta}{2}}\geq \frac{1}{4\alpha}\E^{-2\alpha T}\delta \Big]
\xrightarrow{\eps \downarrow 0}-\infty.
\end{align*}
Next, since $V^\eps$ is exponentially equivalent to $\psi$ with speed $\eps$, then $\lim_{\eps\downarrow 0}\eps\log \PP[A_\eps^c] = -\infty$. Finally, looking at the LDP with speed $\eps^{1 + \frac{\beta}{2}}$, $V^\eps$ is still  exponentially equivalent to $\psi$ by the same argument as before, thus similarly $\eps^{\frac{1}{2} + \frac{\beta}{4}}\int_{0}^{\cdot} g(V^\eps_t)\D W_t = \eps^{\frac{\beta}{4}}Y^\eps$ satisfies an LDP with good rate function having a unique minimum at zero and with speed $\eps^{1 + \frac{\beta}{2}}$. Therefore, $\eps^{\frac{\beta}{4}}Y^\eps$ is exponentially equivalent to zero and
$\lim\limits_{\eps\downarrow 0}\eps \log \PP[B_\eps^c] = -\infty$, which completes the proof.

\section{Variations around Varadhan's Lemma}\label{thm:varadhan}
\small{
Varadhan's integral lemma is a generalisation of Laplace's method. It gives the asymptotic behavior of $\EE[\E^{\frac{\varphi(Z_{\eps})}{\eps}}]$ on a log scale for a family of random variables $Z_{\eps}$ and a continuous function $\varphi$. The Laplace's method states that under some conditions the following relation holds
$$\lim_{n\rightarrow\infty}\frac{1}{n}\log \int_a^b \E^{n f(x)} \D x = \sup\limits_{x\in[a,b]}f(x).$$
One notable application of Varadhan's integral lemma is finding a good change of measure in importance sampling, as we will see it in the following sections. For theses applications, we need a slightly more general formulation of the lemma than the one found in~\cite{Dembo2010}. Nevertheless, the proof in~\cite{Dembo2010} can be easily adapted as done in~\cite{Robertson2010}.
\begin{theorem} (Varadhan's Integral Lemma)\label{thm:Varadhan}
Let $\Xx$ be a metric space, $Z^{\eps}$ a family of $\Xx$-valued random variables that satisfies a LDP with good rate function $I:\Xx\rightarrow [0,+\infty]$ and let $\varphi : \Xx \rightarrow \R$ be a continuous function. Assume further either the tail condition
$$
\lim\limits_{M\downarrow \infty} \limsup_{\eps \downarrow 0}
\eps\log \EE\left[\exp\left\{\frac{\varphi(Z^{\eps})}{\eps}\right\}\ind_{ \{ \varphi(Z^{\eps}) \geq M\}}\right]  = - \infty
$$
or the following moment condition for some $\gamma > 1$ (because it implies the previous tail condition)
$$
\limsup\limits_{\eps \downarrow 0}
\eps\log \EE\left[\exp\left\{\gamma \frac{\varphi(Z^{\eps})}{\eps}\right\}\right]
< \infty .
$$
Then
$$
\lim_{\eps \downarrow 0}
\eps\log \EE\left[\exp\left\{\frac{\varphi(Z^{\eps})}{\eps}\right\}\right]
= \sup\limits_{x \in X} \{\varphi(x) - I(x) \}.
$$
\end{theorem}
The below modified theorem was proven in~\cite{Robertson2010} and allows the above function $\psi$ to reach $- \infty$ and accounts for cases where the problem is written in term of a family of functions $\{\psi_{\eps}\}_{\eps>0}$.
\begin{theorem}(Modified Varadhan's Integral Lemma)\label{thm:varadhan_modified}
Let $\Xx$ and $\Yy$ be two metric spaces, $Z^{\eps}$ a family of $\Xx$-valued random variables that satisfies a LDP with good rate function $I:\Xx\rightarrow [0,+\infty]$ and let $\varphi : \Yy \rightarrow [-\infty,+\infty)$ and $\varphi : \Xx \rightarrow \R$  be two continuous functions. Let $\Lambda : \Xx \rightarrow \Yy$ be a continuous map and $\{\Lambda_{\eps} : \Xx \rightarrow \Yy\}$ be a family of measurable functions such that $\Lambda_{\eps}(Z^{\eps})$ is exponentially equivalent to $\Lambda(Z^{\eps})$. Suppose that there exists $\gamma > 1$ such that
$$
\limsup\limits_{\eps \downarrow 0}
\eps\log \EE\left[\exp\left\{\gamma \frac{\varphi(\Lambda_{\eps}(Z^{\eps})) + \varphi(Z^{\eps})}{\eps}\right\}\right]< \infty.
$$
Then
$$
\lim_{\eps \downarrow 0}
\eps\log \EE\left[\exp\left\{\frac{\varphi(\Lambda_{\eps}(Z^{\eps})) + \varphi(Z^{\eps})}{\eps}\right\}\right]
= \sup\limits_{x \in X} \{\varphi(\Lambda(x)) + \psi(x) - I(x) \}.
$$
\end{theorem}
}

\section{Black-Scholes with deterministic volatility}\label{sec:BS}
\small{
We consider here the simple Black-Scholes model with 
$$
\D X_t = -\frac{1}{2}\sigma^2(t) \D t + \sigma(t) \D W_t, \quad X_0 = 0,
$$
where $\sigma\in \Cc_T$ is a deterministic function adhering to Assumption~\ref{ass:SDE_coeffs}. 
The goal of this section is to provide full details of this specific case, 
which we often refer to in the text.
We consider options with payoffs that are continuous functions of~$X$. 
Let $G:\Cc([0,T]\to \RR)$ be a continuous function and $F := \log|G|$, 
which is then a continuous function with values in $[-\infty,+\infty)$. 
Our objective is to find an almost surely positive random variable~$Z$ with $\EE[Z] = 1$ minimising $\EE[G^2(X)Z^{-1}] = \EE\left[\E^{2F(X)}Z^{-1}\right]$.

\subsection{Small-noise}
The small-noise approximation~\eqref{eq:small_noise} reads
$$
\D X_t^{\eps} = -\frac{1}{2}\sigma^2(t)\D t + \sqrt{\eps}\sigma(t) \D W_t, \qquad X_0 ^{\eps} = 0.
$$

\paragraph{\textit{Deterministic change of drift}}
With the change of measure
$$
\left. \frac{\D\QQ}{\D\PP}\right|_{\Ff_T}
 := \exp\left\{-\frac{1}{2}\int_0^T\dot{h}_t^2 \D t  + \int_0^T\dot{h}_t\D W_t \right\},
$$
with $h\in\HH_T^0$, we can write
$$
L(h):=\limsup\limits_{\eps\downarrow 0} \eps \log\EE\left[\exp\frac{1}{\eps}\left\{{2F(X^{\eps}) +\frac{1}{2}\int_0^T\dot{h}_t^2 \D t  - \dot{h}\cT W^{\eps}}\right\}\right].
$$
To apply Varadhan's lemma and estimate $L(h)$, we assume there exists $\gamma >1$ such that 
\begin{equation}\label{eq:BSAssume}
\limsup\limits_{\eps\downarrow 0} \eps\log\EE\left[
\exp\left\{4\gamma \frac{F(X^{\eps})}{\eps}\right\}\right] <\infty.
\end{equation}
This condition is satisfied if, for example, the following assumption holds:
\begin{assumption}\label{ass:payoffGuasoni}
$F:\Cc_T\rightarrow\RR$ with $F(x)\leq K_{1} + K_{2} \sup\limits_{t\in[0,T]}|x_t|^{\alpha}$
for some $\alpha\in (0,2)$, $K_{1}, K_{2}>0$.
\end{assumption}
\begin{lemma}\label{lem:BSVaradhanCond}
If $\dot{h}$ is of finite variation, the function
\[
\Lf(\cdot;h):
x\in \Cc_T \mapsto 2F\left(-\frac{1}{2}\int_{0}^{\cdot} \sigma^2(t)\D t + \int_{0}^{\cdot} \sigma(t) \D x_t\right)
 + \half\left(\int_0^T \dot{h}_t^2\D t-\int_0^T \dot{x}_t^2 \D t\right)
 - \dot{h}\cT x
\]
is well defined and continuous and for every $\gamma >0$, 
$$
\limsup\limits_{\eps\downarrow 0} \eps \log\EE\left[\exp\frac{\gamma}{\eps}\left\{\frac{1}{2}\int_0^T \dot{h}_t^2\D t
- \dot{h}\cT W^{\eps}\right\}\right] \text{ exists in }\RR.
$$
\end{lemma}
\begin{proof}
If $\dot{h}$ is of finite variation and $x$ is continuous with $x(0)=0$, 
then $\int_{0}^{\cdot}\dot{h}_t\D x_t$ is well defined as a Riemann-Stieltjes integral and 
$$
\left|\dot{h}\cT x\right| = \left|\dot{h}(T)x(T) - x\cT \dot{h}\right| \leq \|x\|_\infty\left(|\dot{h}(T)| + \mathrm{TV}(\dot{h})\right ),
$$
where $\mathrm{TV}(\cdot)$ denotes the total variation. 
Similarly, $\int_{0}^{\cdot}\sigma(t)\D x_t$ is well defined and for all $t \in [0,T]$,
$$
\left|\sigma(\cdot)\ct x\right| \leq \|x\|_\infty\left(\|\sigma\|_\infty + \mathrm{TV}(\sigma)\right).
$$
Since $F$ is continuous, the first statement, about existence and continuity, holds.
The second one is a direct consequence of the computation of the exponential Gaussian moment: for every $\gamma >0$, 
$$
\eps  \log\EE\left[\exp\frac{\gamma}{\eps}\left\{\frac{1}{2}\int_0^T \dot{h}_t^2\D t
- \dot{h}\cT W^{\eps}\right\}\right]
= \frac{\gamma(1+\gamma)}{2}\int_0^T\dot{h}_t^2\D t.
$$
\end{proof}
\begin{lemma}
If~$\dot{h}$ is of finite variation, then 
$L(h) = \sup_{x \in \HH_T^0}\Lf(x;h)$.
\end{lemma}
\begin{proof}
Let $\gamma>1$ such that~\eqref{eq:BSAssume} holds. Then by Cauchy-Schwarz and Lemma~\ref{lem:BSVaradhanCond},
$$
\limsup\limits_{\eps\downarrow 0} \eps \log\EE\left[\exp\frac{\gamma}{\eps}\left\{2F\left(-\frac{1}{2}\int_{0}^{\cdot} \sigma^2(t)\D t + \int_{0}^{\cdot} \sigma(t) \D W_t^\eps\right) + \frac{1}{2}\int_0^T \dot{h}_t^2\D t - \int_0^T \dot{h}_t \D W_t^{\eps} \right\}\right] <\infty,
$$
so that the conditions of Theorem~\ref{thm:Varadhan} are verified.
The continuity has already been shown.
\end{proof}
Since the map $\Gg: W\mapsto X$ is continuous by the continuity of the It\^{o} integral, we can freely introduce $\widetilde F = F\circ\Gg$ and the existence of a minimum to the dual version of 
$\inf\limits_{h \in \HH_T^0} \sup\limits_{x \in \HH_T^0} \Lf(x;h)$,
namely $\sup\limits_{x \in \HH_T^0} \inf\limits_{h \in \HH_T^0}\Lf(x;h)$,
can be proved as in~\cite{Guasoni2007} under Assumption \ref{ass:payoffGuasoni} by choosing $M=0$ in \cite[Lemma~7.1]{Guasoni2007}. 
The minimum is then attained  for $h=x$ and equal to 
\begin{equation}\label{pb-1}
    \sup\limits_{x \in \HH_T^0}2F\left(-\frac{1}{2}\int_{0}^{\cdot} \sigma^2(t)\D t + \int_{0}^{\cdot} \sigma(t) \dot{x}_t \D t\right)
 - \int_0^T \dot{x}_t^2 \D t\,.
\end{equation}
Furthermore, it immediately follows from \cite[Theorem~3.6]{Guasoni2007} that, if $h^*\in \HH_T^0$ is of finite variation and is a solution to~\eqref{pb-1}, then it is asymptotically optimal if
$$
L(h^{*}) = 2F\left(-\frac{1}{2}\int_{0}^{\cdot} \sigma^2(t)\D t + \int_{0}^{\cdot} \sigma(t) \dot{h}^{*}_t \D t\right) - \int_0^T |\dot{h}^{*}_t|^2 \D t.
$$

Therefore, in order to derive a change of measure, we search for~$h^{*} \in \HH_T^0$ such that
$$
h^{*} = \argmax\limits_{x \in \HH_T^0}F\left(-\frac{1}{2}\int_{0}^{\cdot} \sigma^2(t)\D t + \int_{0}^{\cdot} \sigma(t) \dot{x}_t \D t\right)
 - \frac{1}{2}\int_0^T \dot{x}_t^2 \D t.
$$

\paragraph{\textit{Simplified deterministic change of drift}}
We consider a simplified version of the problem. 
Since $\eps\int_{0}^{\cdot}\sigma^2(t)\D t$ is exponentially equivalent to $0$, the results of the previous section remain valid 
when  replacing $F(-\half\int_{0}^{\cdot} \sigma^2(t)\D t + \int_{0}^{\cdot} \sigma(t) \dot{x}_t \D t)$ with $F(\int_{0}^{\cdot} \sigma(t) \dot{x}_t \D t)$. 
The problem then becomes
$$
h^{*} = \argmax\limits_{x \in \HH_T^0} F\left(\int_{0}^{\cdot} \sigma(t) \dot{x}_t \D t\right) - \frac{1}{2}\int_0^T \dot{x}_t^2 \D t.
$$

\subsubsection{Options with path-dependent payoff.}
Consider the payoff $G(\alpha\cT X)$ with $G:\RR \to \RR$ differentiable and~$\alpha$ 
a positive continuous function of finite variation. 
The optimisation problem reads
$$
h^{*} = \argmax\limits_{x \in \HH_T^0} F\left(\int_0^T \alpha_t \sigma(t) \dot{x}_t \D t\right) - \frac{1}{2}\int_0^T \dot{x}_t^2 \D t,
$$
with F = $\log|f|$ in the price small-noise approximation setting
and $F = \log\left|G(\cdot - \frac{1}{2}\int_0^T\alpha_t \sigma^2(t) \D t)\right|$ in the log-price small-noise approximation one. 
The mapping $\dot{y}_t = \alpha_t\sigma(t)\dot{x}_t$
yields the Euler-Lagrange equation
$
\frac{\D}{\D t}\frac{\dot{y}_t}{(\alpha_t\sigma(t))^2} = 0$,
or in terms of the original problem
$\frac{\D}{\D t}\frac{\dot{x}_t}{\alpha_t\sigma(t)} = 0$.
We thus search for solution of the form $\dot{x}_t = \beta \alpha_t\sigma(t)$ for some $\beta \in\RR$. 
The optimisation problem then reads
\[
\begin{cases}
\beta^{*} &= \displaystyle \argmax\limits_{\beta \in\RR} F\left(\beta \int_0^T (\alpha_t \sigma(t))^2 \D t\right)  - \frac{1}{2}\beta^2 \int_0^T (\alpha_t \sigma(t))^2 \D t,\\
h^{*} &= \displaystyle \beta^{*} \int_{0}^{\cdot} \alpha_t \sigma(t) \D t.
\end{cases}
\]
For example, if $F(x) = \log(\E^x - \E^c)^+$ for $x \in\RR$ and $c \in\RR$, $\widetilde{\beta}$ will be the unique solution on $(1,\infty)$ of
$v\beta + \log(\beta - 1) - \log(\beta) - c = 0$.
\begin{example}\ 
\begin{itemize}
\item \textup{(European Call option):} $f(x) = (\E^{x} - K)^+$ and $\alpha_t=1$.
\item \textup{(Geometric Asian Call option):} $f(x) = (\E^{x} - K)^+$ and $\alpha_t=1 - t/T$. 
\end{itemize}
\end{example}

\subsection{Small-time}
The small-time approximation~\eqref{eq:small_time} reads
$\D X_t^{\eps} = -\frac{1}{2}\eps\sigma^2(\eps t)\D t + \sigma(t\eps) \D W_t^{\eps}$,
with $X_0 ^{\eps} = 0$ and $W^{\eps} = \sqrt{\eps}W$.
The couple $(\sigma(t\eps), W_t^\eps)$ is exponentially equivalent to $(\sigma(0),W_t^\eps)$. 
Since~$\IIW$ is the good rate function of the LDP verified by $(W^\eps)$ by Schilder's Theorem~\cite[Theorem 5.2.3]{Dembo2010}, $(\sigma(t\eps), W_t^\eps)$ verifies a LDP with good rate function $\II(s,w) = \IIW(w) + \infty \ind_{\{s\neq \sigma(0)\}}$. 
By Theorem~\ref{thm:robertsonLDP}, $(\int_{0}^{\cdot}\sigma(t\eps)\D W_t^\eps,W^\eps)$ satisfies a LDP with the same good rate function as $(\int_{0}^{\cdot}\sigma(0)\D W_t^\eps,W^\eps)$. 
Noticing that  $\int_{0}^{\cdot} \eps \sigma(t\eps)^2\D t$ is exponentially equivalent to zero, our method then leads to the same solution as the problem for 
$\D X_t^{\eps} = -\frac{1}{2}\eps\sigma^2(0)\D t + \sigma(0) \D W_t^{\eps}$,
which was treated above. 
In this small-time setting, we lose all information on the path of~$\sigma$, except for its initial value.

\subsection{Numerical results}
We provide numerical evidence in the Black-Scholes model with $S_0=50$ in the log-price small-noise approximation.
In order to compare estimators, we look at Asian Arithmetic Call options, 
with payoff $(\frac{1}{T}\int_0^T \exp(X_t) \D t - K)^+$. The form of the solution of the optimisation problem studied previously can be found in~\cite{Guasoni2007}. We compare the naive Monte-Carlo estimator, the antithetic Monte-Carlo estimator, 
the control estimator based on the price of Geometric Asian options 
(with payoff $\{\exp[\frac{1}{T}\int_0^T X_t \D t] - K\}^+$), 
that can be computed explicitly, 
and the LDP-based importance sampling estimator above.
Instead of simulating $\int_0^T \exp(X_t)\D t$, we consider a discretised payoff on $n=252$ dates and draw $10^5$ paths.
For the LDP-based estimator, the law of $W$ after the change of measure is given by Girsanov theorem.
In what follows, when we refer to variance reduction we mean the ratio of variance of the classical Monte-Carlo estimator over the variance of estimator in question.
As we can see in Table~\ref{table:BSCompare} and in Figure~\ref{fig:BSLDP}, 
even in non-rare events, the LDP estimator provides good variance reduction. 
However, it is mainly in the context of rare events that it performs best and outperforms the other estimators (Figure~\ref{fig:BSCompare} and Table~\ref{table:BSCompare}), 
revealing the true power of LDP-based importance sampling estimators.

\begin{figure}
    \centering
    \includegraphics[width=0.4\textwidth]{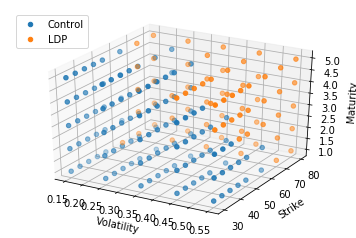}
    \includegraphics[width=0.4\textwidth]{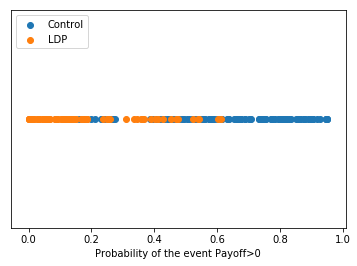}
    \caption{\small{Left: Comparison of the best estimator (in terms of variance) between the control and the LDP one for different values of $\sigma, K, T$.
    Right: Same but plotted againts the probability of a positive payoff (computed using the estimated prices of the options with the values of $\sigma, K, T$ from the left.)
    }}
    \label{fig:BSCompare}
\end{figure}
\begin{figure}
    \centering
    \includegraphics[width=0.4\textwidth]{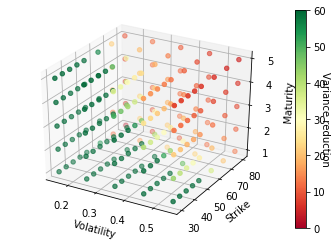}
    \includegraphics[width=0.4\textwidth]{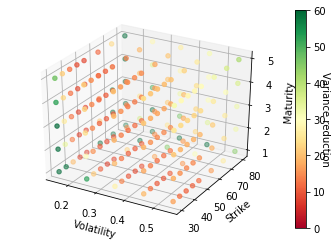}
    \caption{\small{Variance reduction of the control (left) and 
     the LDP (right) estimators}
     }
    \label{fig:BSLDP}
\end{figure}

\begin{table}
\centering
\begin{tabular}{lrrrr}
\toprule
 Strike $K$ &  Antithetic &  Control &  LDP &  Probability of positive Payoff \\
\midrule
     30 &          64 &      769 &   53 &                            0.95 \\
     35 &          59 &      775 &   21 &                            0.95 \\
     40 &          31 &      744 &   10 &                             0.9 \\
     45 &          10 &      575 &  7.9 &                            0.75 \\
     50 &         3.8 &      336 &  8.6 &                             0.5 \\
     60 &         2.2 &       69 &   22 &                            0.11 \\
     70 &         2.0 &       16 &  123 &                           0.013 \\
     80 &         2.3 &      6.9 & 1445 &                          0.001 \\
\bottomrule
\end{tabular}
\caption{\small{Variance reduction for several estimators for Arithmetic Call options in Black-Scholes with $S_0=50$, $r=0.05$, $\sigma = 0.25$ and $T=1$.}}
\label{table:BSCompare}
\end{table}

}
\end{document}